\documentclass[12pt]{article}
\usepackage{authblk}
\usepackage{geometry}
\usepackage{amsmath}
\usepackage{amssymb}
\usepackage{amsthm}
\usepackage{mathrsfs}
\usepackage{subcaption}
\usepackage{customcommands}
\usepackage{bm}
\usepackage{Nettastyle}
\usepackage{thmtools}
\usepackage{thm-restate}

\usepackage{physics} 
\usepackage{tensor}
\usepackage{tcolorbox}

\usepackage[normalem]{ulem}

\newcommand{\CV}{\mathcal{C}_V} 
\newcommand{\dCV}{\dot{\mathcal{C}}_V} 
\newcommand{\CF}{\mathcal{C}_F} 

\preprint{MIT-CTP/5329}

\title{General Bounds on Holographic Complexity}

\author{Netta Engelhardt}
\author{and \AA{}smund Folkestad}
\emailAdd{engeln@mit.edu}
\emailAdd{afolkest@mit.edu}
\affiliation{Center for Theoretical Physics, Massachusetts Institute of Technology, \\Cambridge, MA 02139, USA}

\abstract{We prove a positive volume theorem for asymptotically AdS spacetimes: the  maximal volume
    slice has nonnegative vacuum-subtracted volume, and the vacuum-subtracted volume vanishes if and
    only if the spacetime is identically pure AdS. Under the Complexity=Volume proposal, this constitutes a
        positive holographic complexity theorem. The result features a number of parallels with the positive energy
        theorem, including the assumption of an energy condition that excludes false vacuum decay (the AdS weak energy condition). Our
        proof is rigorously established in broad generality in four bulk dimensions, and we provide
        strong evidence in favor of a generalization to arbitrary dimensions.  Our techniques also yield
        a holographic proof of Lloyd's bound for a class of bulk spacetimes.  We further establish a
        partial rigidity result for wormholes: wormholes with a given throat size are more complex than
        AdS-Schwarzschild with the same throat size. 
    }

\begin{document}

    \maketitle

    \section{\label{sec:intro}Introduction}
    The interpretation of aspects of spacetime geometry in terms of quantum information theoretic
    quantities -- from (quantum) extremal surfaces \cite{RyuTak06,HubRan07,Wal12,FauLew13,Don13,EngWal14, BroGha19, EngPen21, EngPen21b} to holographic circuit complexity \cite{StaSus14, RobSta14, SusZha14, Ali15,BroRob15, BroRob15b, LehMye16, CouFis16, ChaMar16, CarMye16, CarCha17}  has been a critical aspect of recent progress towards a complete holographic description
    of the black hole interior. The Complexity=Volume
    dictionary entry \cite{StaSus14}, which is our focus in this article, was initially motivated in part by the potential of the maximal volume slice to probe the deep black hole interior. The proposal relates the circuit complexity $\mathcal{C}$ of preparing the CFT
    state $\ket{\psi(\tau)}$ from some reference state $\ket{R}$ and the bulk maximal volume slice $\Sigma_{\tau}$ anchored at $\tau$: \footnote{Here we can either take the perspective where postselection is allowed, i.e.\ non-unitary
        gates are allowed, or we can take this to compute the size of the
    circuit before including any postselection, and correct this by including the python's
    lunch~\cite{BroGha19}.}
    \begin{equation}\label{eq:CVdef}
    \begin{aligned}
        \mathcal{C}(\ket{\psi(\tau)}) = \max_{\Sigma_{\tau}}\frac{
            \vol[\Sigma_{\tau}] }{ G_N \ell} \equiv
        \mathcal{C}_V,
    \end{aligned}
    \end{equation}
    where $\ell$ is a reference length scale which we take to be the AdS radius
    $L$, $\vol[\Sigma_{\tau}]$ is the
    cutoff-regulated volume of the hypersurface $\Sigma_{\tau}$, and the maximum is
    taken over all bulk timeslices anchored at $\tau$ in the dual spacetime $(M, g)$. Here we have given the name $\mathcal{C}_{V}$ to the geometric quantity for convenience. This proposal has led to numerous
    insights on spacetime emergence in the deep interior, from the late-time behavior of wormhole
    volumes \cite{StaSus14, CarCha17}, to the connection between bulk depth and momentum~\cite{
    SusZha20,BarMar19, BarMar20, BarMar20b}, to the late-time breakdown of
    classical GR \cite{Sus20}, among others (e.g.~\cite{BouFef19}).
    
        Several aspects of the CV proposal remain ambiguous: the appropriate
    reference length scale $\ell$, a complexity distance measure,\footnote{That is, a
    set of allowed unitary gates, or in the case of Nielsen geometry \cite{NieDow06}, a cost function.} and a choice of reference
    state $\ket{R}$. It is natural to consider the vacuum $\ket{0}$ as reference state, but since $\mathcal{C}_V(\ket{0})$ is
    non-vanishing this is not feasible. However, we could instead identify
    $\mathcal{C}(\ket{\psi}, \ket{0})$ with the so-called complexity of formation 
    $\mathcal{C}_F(\ket{\psi})$ \cite{BroRob15b, ChaMar16}, which is just the
    vacuum-subtracted version of $\mathcal{C}_V$:
    \begin{equation}
    \begin{aligned}
        \mathcal{C}_F(\ket{\psi(\tau)}) \equiv \mathcal{C}_V(\ket{\psi(\tau)})- \mathcal{C}_V(\ket{0})
        = \max_{\Sigma_{\tau}} \frac{ \vol[\Sigma_{\tau}] - n \vol[\Sigma_{\rm
        AdS}] }{ G_N \ell },
    \end{aligned}
    \end{equation}
    where $n$ is the number of disjoint conformal boundaries on which the CFT lives.
        For sufficiently fast matter falloffs this quantity is finite
        as near-boundary cutoff is removed
        (see Sec.~\ref{sec:lowerbound}). 
    
    It is
    immediately clear that this interpretation of $\mathcal{C}$ requires a highly nontrivial geometric consistency check: 
    \be \label{eq:CVpositivity}
    \CF(\ket{\psi}) \geq 0,
    \ee
    where equality holds if and only if
    $\Sigma_{\tau}$ is identically a maximal volume slice
    of pure AdS. The positivity of $\CF$ has been
    investigated in particular spacetimes~\cite{ChaMar16,FuMal18} and
    perturbations thereof~\cite{FloMie18,BerGal20}. 

    The statement~\eqref{eq:CVpositivity} is reminiscent of the positive energy
    theorem~\cite{SchYau79,SchYau81,Wit81,Wan01, ChrHer01, AndCai07, ChrGal21}, which
    in rough terms states that the mass of a spacetime with a well-behaved and complete initial data set
    satisfying the dominant energy condition is nonnegative; furthermore, it vanishes if and only if the
    data is identically the vacuum.  The positive energy theorem assumes either asymptotic
    flatness \cite{SchYau79,SchYau81,Wit81}, or asymptotic hyperbolicity \cite{Wan01, ChrHer01, AndCai07,ChrGal21}, and in the latter case it is sufficient to assume
    the weak energy condition.

    Is there an analogous equally general theorem that applies to ${\cal C}_{F}$ under the assumptions of regularity of the
    initial data and some energy condition? Here we prove rigorously for a broad class of spacetimes that the answer is yes, and that the relevant energy
    condition is the AdS weak curvature condition (to which we refer as the WCC for short):
    \begin{equation}
        \begin{aligned}\label{eq:WCC}
        \mathcal{E} \equiv t^a t^b \left(R_{ab} - \frac{ 1 }{ 2 }g_{ab}R - \frac{ d(d-1) }{ 2 L^2 }g_{ab} \right) \geq 0, \qquad
        \forall \text{ timelike } t^a,
    \end{aligned}
    \end{equation}
    which simplifies to the usual weak energy condition  $T_{ab}t^a t^b \geq 0$ in Einstein gravity
    (under an appropriate shift of a constant from the matter Lagrangian to the gravitational
    Lagrangian). Our positive (holographic) complexity theorem then states: 

    \vspace{0.4cm}

    \noindent \textit{Let $\Sigma$ be a set of regular and complete asymptotically AdS$_{4}$ initial data on a maximal volume slice
  anchored at a static slice of the conformal boundary,
    and with $n$ asymptotic boundaries. If $\Sigma$ satisfies the WCC, then it has nonnegative vacuum-subtracted volume:}
    $$\vol[\Sigma]-n \vol[\Sigma_{\mathrm{AdS}}]\geq 0, $$
    \textit{provided each outermost minimal surface on $\Sigma$ is connected.}
    \textit{Equality holds if and only if $\Sigma$ is a static slice of pure AdS}. 

    \vspace{0.4cm}

    \noindent Thus, among all asymptotically AdS$_{4}$ spacetimes satisfying the WCC and our assumption on minimal
    surfaces, the vacuum is
    strictly the simplest: like the positive energy theorem, our positive complexity theorem provides a
    new vacuum rigidity result. 
    
    The theorem is proved in Sec.~\ref{sec:proofs} via the application of a powerful mathematical result by Brendle and
    Chodosh~\cite{BreCho14, Cho14} that grants control over volumes of asymptotically hyperbolic Riemannian manifolds. 
    For technical reasons, the result applies for maximal volume slices where each outermost minimal surface (see
    Sec.~\ref{sec:lowerbound}) is connected, effectively restricting to spacetimes where for any connected component
    $\mathscr{I}$ of the conformal boundary, $\partial J^{-}[\mathscr{I}]\cap \Sigma$ is connected on the maximal volume slice $\Sigma$.

    Is the WCC actually necessary? Let us address this by taking inspiration from the positive energy
    theorem. Positivity of the mass and rigidity of the vacuum is indeed violated by false vacuum decay,
    which violates the dominant energy condition. We may expect a similar phenomenon here, as the WCC
     serves the same purpose here as the dominant energy condition: it
    excludes false vacuum decay and negative local energy densities. 
    Indeed, existing violations of positivity of ${\cal C}_{F}$ \cite{ChaGe18,BerGal20} do indeed violate our assumptions; we explore this in more detail in a companion article~\cite{EngFol21b}, where we will demonstrate that matter fields (e.g. scalar
    tachyons) violating the WCC  -- but satisfying the Null Energy Condition $T_{ab}k^{a}k^{b}\geq 0$ -- can indeed violate the positive complexity theorem, just as they can
violate the positive energy theorem \cite{HerHor04b}. Perhaps more
surprisingly, it also turns out that the inclusion of compact dimensions can violate positivity of ${\cal C}_{F}$, even when the WCC holds. Common to all of the violations in
\cite{EngFol21b} is the presence of VEVs for relevant CFT scalar primaries, and
so the results proven in this paper could in more general theories likely be
viewed as constraining how certain subsets of operators affect $\mathcal{C}_V$. 

What about our restriction to four dimensions? Unlike the WCC or absence of
nontrivial compact factors, we expect that this restriction can
be lifted. The generalization to arbitrary dimensions can be proven from a generalization of a
well-known mathematical conjecture by Schoen~\cite{Schoen};
without relying on conjectures, we give a proof under the assumption of spherical or planar
symmetry that $\CF$ is positive (again assuming the WCC). We also highlight an existing
mathematical theorem \cite{HuJi16} proving that sufficiently small (but finite) WCC-preserving deformations of pure AdS
result in positive $\CF$, thus establishing in all
dimensions that pure AdS is a local minimum of $\CF$ within the space of WCC-preserving
asymptotically AdS spacetimes. We consider the combination of our results, Schoen's conjecture and the
small-deformation results \cite{HuJi16} to be strong evidence in favor of the higher-dimensional generalization of the
positive complexity theorem. 

Rigidity of the vacuum implies that the vacuum has the lowest complexity of formation -- it is simplest, under assumption
of $\mathcal{C}\sim \vol[\Sigma]$; there is an analogous expectation that the thermofield double is a fairly simple state that minimizes holographic complexity
within some class. The corresponding AdS dual would be a proof that the
Schwarzschild-AdS geometry minimizes the volume within some class of wormholes
(see related work by~\cite{SheSta14}).    Such a result about ${\cal C}_{F}$ would in fact constitute a parallel with a less famous companion to the positive
energy theorem, the Riemannian Penrose Inequality, a lower bound on the mass given the existence of a minimal area surface \cite{HuiIlm01, Bra99,
Bra07, HusSin17}. More precisely, the area of a minimal surface on a maximal Cauchy slice in a spacetime of a given mass is conjectured to
be identical to that of Schwarzschild with the same mass if and only if the spacetime is identical to Schwarzschild.\footnote{Note that while some results are
established for the Riemannian Penrose Inequality in AdS~\cite{DahGic12,LopGir12,GeWan13,LeeNev13, HusSin17, Mar09}, the
full statement remains conjectural. We bring it up here to draw a parallel.} 

\begin{figure}
\centering
\includegraphics[width=0.7\textwidth]{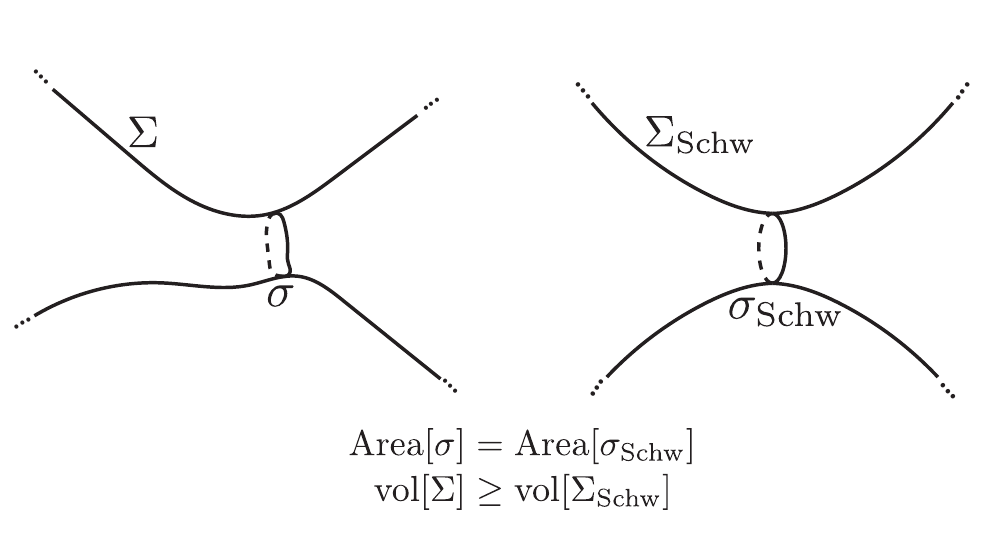}
\caption{Illustration of the volume comparison theorem at fixed throat area.
$\Sigma_{\rm Schw}$ is a static slice of Schwarzschild, while $\Sigma$ is a
maximal volume slice anchored at a static boundary time.}
\label{fig:wormholecomp}
\end{figure}

Can our techniques provide a rigorous backing to this expectation that Scwarzschild-AdS is simplest in some precise
sense? The analogy with the positive energy theorem would suggest that the answer is yes, and indeed it is: the same methodology used to prove ${\cal C}_{F}\geq 0$ yields a holographic complexity parallel of
the Penrose Inequality (not restricted to a moment of time symmetry). 
\vspace{0.4cm}

\noindent \textit{Let $\Sigma$ be a complete maximal volume slice in an asymptotically AdS$_{4}$
    spacetime
with two disconnected asymptotic boundaries, equipped with initial data
satisfying the WCC, and anchored at a static boundary slice. If $A$ is the
area of an outermost minimal-area surface on $\Sigma$, then}
$$\vol[\Sigma]-\vol[\Sigma_{\mathrm{Schw}}]\geq 0, $$
\textit{provided each outermost minimal surface on $\Sigma$ is connected.}
\textit{Equality holds if and only if $\Sigma$ is identical to $\Sigma_{\mathrm{Schw}}$, the static slice of AdS-Schwarzschild with throat area $A$.}

\vspace{0.4cm}

\noindent That is, the area of the throat fixes the lowest possible complexity (as computed
by the volume), which is that of Schwarzschild with the same throat area.
See Fig.~\ref{fig:wormholecomp} for an illustration. In general dimensions we prove this under the restriction of spherical or planar
symmetry. 
Once again, as we show in our companion
paper~\cite{EngFol21b}, the WCC is necessary: it is possible to construct classical wormhole solutions
satisfying the NEC with negative $\CF$.

Finally, we make use of the tools developed in
our proof of the positive
complexity theorem to investigate a closely-related conjecture regarding the growth of complexity:
Lloyd's bound \cite{Llo00}. This bound proposes that the rate of computation of a quantum system is bounded by an expression
proportional to energy, and in the form discussed in \cite{BroRob15, BroRob15b}, it reads
\begin{equation*}
\begin{aligned}
\frac{ \dd \mathcal{C} }{ \dd \tau } \leq c E,
\end{aligned}
\end{equation*}
for some constant $c$. A large portion of the holographic complexity literature is dedicated to the
investigation of this bound in specific spacetimes. In this work, for a class of
solutions, we establish a similar result rigorously. To be precise, in minimally coupled Einstein-Maxwell-Scalar
theory with spherical symmetry and more than three bulk dimensions (assuming the
WCC), the growth of holographic complexity is bounded by the spacetime mass $M$:
\begin{equation}\label{eq:lloydintro}
    \begin{aligned}
        \frac{ 1 }{ G_N L }\Big|\frac{\dd}{\dd \tau}\vol[\Sigma_{\tau}]\Big| \leq \frac{ 8\pi M }{ d-1 }f(M), \qquad 
        f(M) = \begin{cases}
        1 & M \leq \hat{M} \\
        1 + 2 \left(\frac{ M }{ \hat{M} }\right)^{\frac{ 1 }{ d-2 }} & M > \hat{M}
    \end{cases},
\end{aligned}
\end{equation}
where $\hat{M}$ is a mass scale near the Hawking-Page transition (see \eqref{eq:fexplicit} for the specific value), 
and $\Sigma_{\tau}$ a maximal volume slice anchored at round sphere on the boundary -- i.e. at
constant $\tau$ in the boundary metric $ - \dd \tau^2 + L^2\dd \Omega^2$.
The takeaway is that late time holographic complexity growth is at
most linear,\footnote{More precisely, there is a linearly growing upper bound
on holographic complexity.} with
the slope restricted by the mass. Since \cite{Aar16} showed that circuit complexity has a linearly growing upper bound in time (with an appropriate definition of complexity and time), this provides significant evidence for consistency of the CV proposal.

\begin{figure}
\centering
\includegraphics[width=0.5\textwidth]{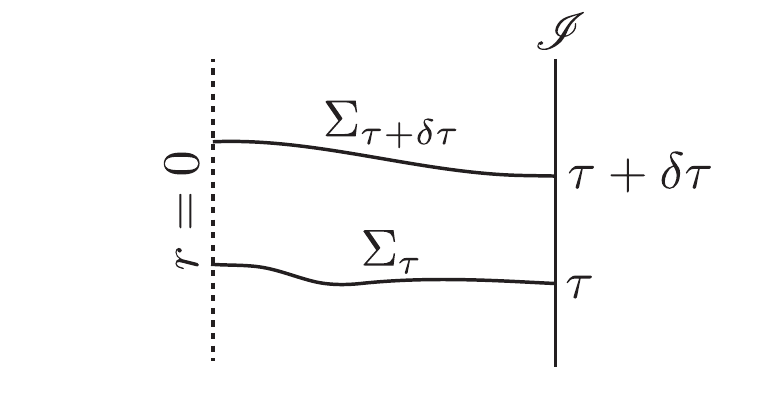}
\caption{Illustration of two static anchored maximal volume slices separated by
boundary static time $\delta \tau$. \eqref{eq:lloydintro} provides a speed limit
on the growth of $\frac{\dd\vol[\Sigma_{\tau}]}{\dd \tau}$.}
\label{fig:sliceperturb}
\end{figure}

The paper is structured as follows. In Sec.~\ref{sec:lowerbound} we state our proven
lower bounds on complexity and proven Penrose inequality. The proofs are then given in a separate subsection, which can be skipped
by readers only interested in the primary results, with the caveats that (1) the technology
developed there reappears in the proof of Lloyd's bound and (2) some of the results
of the proof section are interesting on their own right. 
In Sec.~\ref{sec:lloyd} we discuss upper bounds on complexity growth, with a proof of \eqref{eq:lloydintro}. 
We also give a novel formula for complexity growth in spacetimes with spherical or
planar symmetry and with matter with compact spatial support (in the radial direction). This formula
is valid independent of the WCC and the theory considered, although when assuming the former it gives a novel upper bound on complexity
growth. Finally, in Sec.~\ref{sec:discussion} we discuss implications of our result.
A number of explicit computations are relegated to an appendix.

\section{\label{sec:lowerbound}The Positive Complexity Theorem} 
In this section, we prove the positive volume (complexity) theorem for AdS$_{4}$, state the necessary conjecture for the
proof of the statement in AdS$_{d+1\neq 4}$, and prove the general $d$ case for certain classes of bulk spacetimes. We
also state our Penrose inequality for spatially symmetric spacetimes. The reader only interested in the results may skip
section~\ref{sec:proofs}, which consists of
the proofs.

\subsection{The Vacuum is Simplest}\label{sec:vacsimple}
We work with maximal volume slices $\Sigma$ anchored on static boundary time slices: that is, in the conformal frame where the boundary manifold has coordinates
\begin{equation}
    \dd s^2 = - \dd \tau^2 + L^2 \dd Q^2, 
    \end{equation} 
with $\dd Q^2$ the metric of a maximally symmetric spacelike manifold (i.e. a sphere, plane, or hyperbolic space), 
we take our maximal volume slices to be anchored at constant $\tau$, which we shall refer to as
\textit{static-anchored}; we further take $\Sigma$ to be a manifold without boundaries or
singularities -- a \textit{complete} Cauchy slice.\footnote{$\Sigma$ will of course have a boundary in the conformal completion. The absence of singularities means
that the intrinsic geometry of $\Sigma$ is geodesically complete, which rules out intrinsic
curvature singularities on $\Sigma$ (unless they are distributional). Geodesic
completeness on $\Sigma$ is equivalent to the statement that all closed and bounded sets on $\Sigma$ are compact~\cite{HopRin31}.}
In particular, this means that our results do not apply to spacetimes terminating at
end-of-the-world branes, in which Cauchy slices are manifolds with boundary.

Next, in order to discuss volume-regularization and assumptions on falloffs, it is useful to
recall that for any choice of representative of the boundary conformal structure \cite{FefGra85, GraLee91,
GraWit99}, the 
Fefferman-Graham gauge gives a unique coordinate system in a neighborhood of $\mathscr{I}$:
\begin{equation}
\begin{aligned}
\dd s^2 = \frac{ L^2 }{ z^2 }\left[\dd z^2 + \gamma_{\mu\nu}(z, x) \dd x^{\mu} \dd x^{\nu}\right],
\end{aligned}
\end{equation}
where the conformal boundary is at $z=0$, and where $\gamma_{\mu\nu}(z=0, x)$ is the chosen
representative of the boundary conformal structure. The resulting Fefferman-Graham expansion reads
\begin{equation}
\begin{aligned}
    \gamma_{\mu\nu}(z, x) = \gamma_{\mu\nu}^{(0)}(x) + z^2\gamma_{\mu\nu}^{(2)}(x) + \ldots +
    \gamma_{\mu\nu}^{(d)}z^{d} + z^{d}\log z\bar{\gamma}^{(d)}_{\mu\nu} + \ldots.
\end{aligned}
\end{equation}
We will assume that matter falloffs are sufficiently slow so that the $\gamma_{\mu\nu}^{(n)}$ for $n<d$,
under the equations of motion, only depend on the boundary geometry $\gamma^{(0)}_{\mu\nu}$ and its
derivatives. This is the case for pure Einstein-gravity with a negative cosmological constant, and
also upon the inclusion various standard forms of matter appearing in holography \cite{deHSolSke00}
(provided we only turn on normalizable modes). We will call a spacetime asymptotically AdS (AAdS) whenever it has the above falloffs. This ensures that (1) a finite gravitational mass can be defined from the asymptotic value of
the metric (as in \cite{BalKra99b}), and (2) the divergence structure of maximal volumes slices is universal and coincident with that of a static slice of pure AdS. Once the boundary conformal
frame is fixed, this provides a canonical way to regularize
the volume: cut off the spacetime at $z=\epsilon$ for some small $\epsilon>0$. After taking
the volume difference with the regularized volume of a static slice of pure AdS, the
result remains finite as $\epsilon\rightarrow 0$. 

Finally, before turning to our results, we introduce one additional definition.
We will say that a compact surface $\sigma$ on $\Sigma$ is \textit{outermost minimal} if (1) $\sigma$
divides $\Sigma$ into two connected components (with one of them possibly trivial):
$\Sigma_{\mathrm{in}}$ and $\Sigma_{\mathrm{out}}$, where $\Sigma_{\mathrm{out}}$ intersects the conformal boundary
$\mathscr{I}$ in the conformal completion on a complete Cauchy slice of a single connected component
of $\mathscr{I}$, 
(2) $\sigma$ is a local minimum of the area functional on $\Sigma$, and (3) there is no other
surface satisfying (1) and (2) that is contained in $\Sigma_{\mathrm{out}}$. We note that (1) outermost minimal
surfaces have the property that every surface in $\Sigma_{\mathrm{out}}$ has area greater than $\sigma$ \cite{Bra99} and
(2) such surfaces always exist in wormhole geometries (see proof of Theorem~\eqref{thm:mainthm_d3}).

Our first theorem is a lower bound on maximal volume in four bulk dimensions. 
\begin{thm}\label{thm:mainthm_d3}
Let $\Sigma$ be a static-anchored complete maximal volume slice of an asymptotically globally\footnote{By this we mean that the bulk approaches global AdS$_{4}$ at the asymptotic boundary, so that $Q$ is a sphere.} 
AdS$_{4}$ spacetime $(M,g)$ satisfying the WCC~\eqref{eq:WCC}. If the conformal boundary of $(M, g)$ has $n$ connected
    components $\mathscr{I}_i$, and the outermost minimal surface in $\Sigma$ with respect to each $\mathscr{I}_i$
    is connected (if non-empty), then the volume of $\Sigma$ is bounded from below by the volume of $n$ static slices $\Sigma_{\mathrm{AdS}_{4}}$ of pure AdS$_{4}$:
\be
\vol[\Sigma]-n \vol[\Sigma_{\mathrm{AdS}_{4}}]\geq 0,
\ee 
    with equality if and only if $(M,g)$ is identical to pure AdS$_{4}$ on $\Sigma$. 
\end{thm}
\noindent As mentioned above, while the volumes are formally infinite, comparison with the vacuum is well defined and finite; the vacuum-subtracted volume is finite and positive as the regulator is removed. 

Without the holographic interpretation of the vacuum-subtracted volume as the complexity of
formation, this is a positive (renormalized) volume theorem. With ${\cal C}\propto \vol[\Sigma]$
with some fixed proportionality constant, however, we have the desired positive complexity result: given a state $\ket{\psi
(\tau)}$ of a holographic CFT$_{3}$ on $n$ copies of the static cylinder, the complexity of
$\ket{\psi(\tau)}$ is bounded from below by the complexity of the vacuum $\ket{0}$ on $n$ disjoint spheres: 
    \begin{equation}
    \begin{aligned}
        \mathcal{C}(\ket{\psi}) \geq \mathcal{C}(\ket{0}^{\otimes n}),
    \end{aligned}
    \end{equation}
    with equality if and only if $\ket{\psi(\tau)}=\ket{0}^{\otimes n}$. The complexity of formation is
    guaranteed to be positive for any bulk spacetime satisfying the WCC, and with a connected outermost minimal
    surface. The latter will be true for a single $n$-sided black hole (i.e.\ a wormhole when $n\geq 1$).

As noted in Sec.~\ref{sec:intro}, the rigidity statement of the above theorem bears a strong parallel to vacuum
rigidity following from the positive energy theorem:  under assumptions of the latter,
Minkowkski space \cite{SchYau79, SchYau81} is the unique asymptotically flat spacetime with
vanishing mass. Analogous positive energy results have been established for pure
AdS \cite{Wan01, ChrHer01, AndCai07,ChrGal21} under the assumption of the WCC.  From the
statement of Theorem~\ref{thm:mainthm_d3}, it is clear that a (WCC-satisfying) excitation of volume above the pure AdS value is possible if and
only there is a corresponding non-zero mass: so the uniqueness of the minimal mass spacetime ensures
the uniqueness of the minimal volume spacetime.

   For bulk dimensions other than four, the technical results are restricted to either (1) have
   spherical or planar symmetry,\footnote{Meaning that $(M, g)$ can
    be foliated by a family of codimension$-2$ surfaces left invariant by an algebra of isometries equal to that of the
$d-1$ dimensional sphere or plane. The codimension$-2$ submanifolds $\sigma$
invariant under these symmetries are isometric to the sphere or plane, or
quotients thereof.} or (2) consist of small but finite deformations of the vacuum. Our general expectation is nevertheless that the
results of four bulk dimensions continue to hold more generally. However, in the absence of a
general proof to this effect, we state the lower bounds on complexity of formation in general dimensions as
rigorously as possible, starting with spacetimes with spatial symmetry:\footnote{In fact, this theorem would also apply to hyperbolic symmetry if it could be shown analytically that static hyperbolic black holes have
volumes greater than the vacuum. This is easily checked to be true numerically, but given the
reliance on a numerical computation this result is not rigorously
established. }
\begin{thm}\label{thm:mainthm_dall}
Let $\Sigma$ be a static-anchored complete maximal volume slice of an asymptotically AdS$_{d+1\geq
3}$ spacetime $(M,g)$ with spherical or planar symmetry, and assume $(M, g)$ satisfies the WCC. 
If the conformal boundary of $(M, g)$ has $n$ connected components, then the vacuum-subtracted volume of 
$\Sigma$ is nonnegative:
\be
\vol[\Sigma]- n \vol[\Sigma_{\mathrm{AdS}_{d+1}}] \geq 0,
\ee 
with equality if and only if $(M,g)$ is identical to $n$ copies of pure AdS$_{d+1}$ on $\Sigma$. 
\end{thm}
\noindent The upshot of this theorem is that any AAdS spacetime with spherical or planar symmetry (but arbitrary time
dependence) that satisfies the WCC has a positive complexity of formation unless it is exactly
 pure AdS on the maximal volume slice, and thus also in its domain of dependence.\footnote{Note that this result holds even for AAdS spacetimes that do not have the
falloffs assumed throughout this paper, although in this case
$\vol[\Sigma]-n\vol[\Sigma_{\mathrm{AdS}_{d+1}}]$ is positive and divergent as the
regulator is removed. In this case the spacetime in some sense has infinite mass, since the
nonstandard falloffs result in a divergent Balasubramanian-Kraus \cite{BalKra99b} stress tensor.} 
Note that spherically symmetric spacetimes should be compared with global AdS, and planar
with Poincar\'e-AdS.

A known mathematics theorem immediately yields the desired result for perturbations around pure AdS in general dimensions:
\begin{thm}[\cite{HuJi16}]\label{thm:mainthm_dall_pert} 
    Let $\delta g$ be a metric perturbation 
    to global pure AdS$_{d+1}$, where the perturbation approaches zero sufficiently quickly to not
    affect the asymptotics~\cite{HuJi16}. A sufficiently small 
    such $\delta g$ satisfying the WCC cannot increase the volume of static-anchored maximal-volume slices. 
\end{thm}
\noindent In this result $\delta g$ need not be infinitesimal, as long as it is sufficiently small in an appropriate norm \cite{HuJi16}. 

\subsection{A Lower Bound on Wormhole Volume}\label{sec:TFDsimple}
In four bulk spacetime dimensions, we also prove the analogous result about wormhole complexity:
\begin{thm}\label{thm:tfd_d3}
    Let $(M,g)$ be a WCC-satisfying, connected, globally asymptotically AdS$_{4}$ spacetime with two conformal
    boundaries.
    Let $\Sigma$ be a static-anchored complete maximal volume slice in $M$, and let $\sigma$ be the globally
    minimal surface on $\Sigma$.\footnote{Such a surface must exist. See the proof of
    Theorem~\ref{thm:mainthm_d3}.} Then, if the outermost minimal surfaces in $\Sigma$ with respect to each component is connected,
\begin{equation}\label{eq:throatineq}
    \begin{aligned}
        \vol[\Sigma]\geq \vol[\Sigma_{\mathrm{Schw}}],
    \end{aligned}
    \end{equation}
    where $\Sigma_{\mathrm{Schw}}$ is the totally geodesic slice of Schwarzschild-AdS whose bifurcation surface 
    has area given by $\mathrm{Area}[\sigma]$.  
\end{thm}
This theorem resembles a Riemannian Penrose
inequality~\cite{Pen73} for volumes, or alternatively, for complexity: the area of a minimal surface on the maximal
volume slice $\Sigma$ puts a lower bound on the complexity. The observation that Penrose-like
inequalites exist for volume was first made in \cite{Cho14}, in the context of
a result similar to Theorem~\ref{thm:BreCho}, which is the crucial in our proof of Theorems
\ref{thm:mainthm_dall_pert} and \ref{thm:tfd_d3}.

Similar techniques facilitate a proof of the above result in other dimensions under assumption of spatial symmetry. 
We will use the abbreviation Stat$_k$ to refer to the static AdS black hole with metric 
\begin{equation}
\begin{aligned}
    \dd s^2 = - \left(k + \frac{ r^2 }{ L^2 } - \frac{ m }{ r^{d-2} }\right)\dd t^2 + \frac{ \dd r^2
    }{ k + \frac{ r^2 }{ L^2 } - \frac{ m }{ r^{d-2} } } + r^2\dd Q_k^2,
\end{aligned}
\end{equation}
where $\dd Q_k^2$ is the metric of the $d-1$ dimensional unit sphere, plane, or unit hyperbolic
space for $k=1, 0, -1$, respectively. The constant $m$ is proportional to the mass. We include the hyperbolic case, since interesting intermediate
results in the next section apply for this case as well. We will use the term \textit{maximal spatial symmetry} to
refer to spacetimes with either spherical, planar or hyperbolic symmetry.

Under the assumption of maximal spatial symmetry, we prove the following:
\begin{thm}\label{thm:tfd_dall}
    Let $(M,g)$ be a connected asymptotically AdS$_{d+1\geq 3}$ WCC-satisfying spacetime with
    spherical or planar symmetry, and with a conformal boundary $\mathscr{I}$ with two connected components. Let
    $\Sigma$ be a static-anchored complete maximal
    volume slice in $M$, and take $\sigma \in \Sigma$ to be the globally minimal symmetric surface on
    $\Sigma$.\footnote{For a connected two-sided complete slice $\Sigma$ with the given symmetries,
    such a surface clearly exists.} Then
\begin{equation}
    \begin{aligned}
        \vol[\Sigma]\geq \vol[\Sigma_{\mathrm{Stat}}],
    \end{aligned}
    \end{equation}
    where $\Sigma_{\mathrm{Stat}}$ is the totally geodesic slice of $\mathrm{Stat}_k$ (depending on which symmetry $(M, g)$
    has) whose bifurcation surface has area $\mathrm{Area}[\sigma]$. 
\end{thm}

\subsection{A Riemannian Penrose Inequality with Spatial Symmetry}\label{sec:penrose}
Here we present a novel application of our results, which is somewhat tangential to the main subject of the paper. Uninterested readers may choose to skip ahead to the proofs of the theorems in Sec. \ref{sec:proofs}. The main upshot is that our results directly imply the Riemannian Penrose inequality for maximal spatial symmetry. Using $\Omega_{k}$ to
denote the volume of the unit sphere $(k=1)$, plane $(k=0)$ or unit hyperbolic plane $(k=-1)$, we have
\begin{thm}\label{thm:penrose}
    Let $(\Sigma, h_{ab}, K_{ab})$ be a complete asymptotically AdS$_{d+1\geq 3}$ maximal volume
    initial data set respecting the WCC and with maximal spatial symmetry. Let $\sigma$ be the outermost
    stationary symmetric surface (possibly empty, in which case we define $\Area[\sigma]=0$). Then
    \begin{equation}\label{eq:provenpenrose}
    \begin{aligned}
        \frac{ 16\pi G_N }{ (d-1)\Omega_{k} }M \geq k \left(\frac{ \Area[\sigma] }{ \Omega_{k} }\right)^{\frac{ d-2 }{ d-1 }} + \frac{ 1 }{ L^2
        }\left(\frac{ \Area[\sigma] }{ \Omega_{k} }\right)^{\frac{ d }{ d-1 }},
    \end{aligned}
    \end{equation}
    where $k=1, 0, -1$ for spherical, planar, or hyperbolic symmetry, respectively, and with
    equality if and only if $(\Sigma, h_{ab}, K_{ab})$ is a static slice of $\mathrm{Stat}_k$ or pure AdS.
\end{thm}
\noindent This result does not require a moment of time symmetry, but the surface $\sigma$ will
be (anti)trapped than marginally (anti)trapped unless $K_{ab}=0$ at $\sigma$. Also note that while $\Omega_k$ may be formally infinite
(unless we use isometries to compactify the symmetric surfaces), $\Area[\Sigma]/\Omega_{k}$ and $M/\Omega_{k}$ are well
defined.\footnote{A formulation of the non-compact cases ($k=0, -1$) of this inequality to spacetimes without symmetry would
presumably have to involve some background subtraction. For example, this could likely be done for spacetimes that, on a
spatial slice, approach a static planar or hyperbolic black hole outside a compact set in the conformal completion.}

\subsection{Proofs of Theorems}\label{sec:proofs}
\subsubsection*{Proof of Theorems \ref{thm:mainthm_d3} and \ref{thm:tfd_d3}}
To prove Theorem \ref{thm:mainthm_d3}, our main tool is an extant result about volumes of regions of conformally compact
Riemannian $3$-manifolds with certain asymptotics \cite{BreCho14, Cho14}. The proof of Theorem~\ref{thm:mainthm_d3} is constructed by establishing that (1)
 the requisite technical assumptions used in the mathematical literature hold for
static-anchored maximal volume slices $\Sigma$ of AAdS$_4$ spacetimes when the WCC holds, and (2) 
 if $\Sigma$ has $n$ conformal boundaries it also has $n$ disjoint asymptotic regions of the kind
appearing in the result of \cite{BreCho14, Cho14}.

\begin{figure}
\centering
\includegraphics[width=0.7\textwidth]{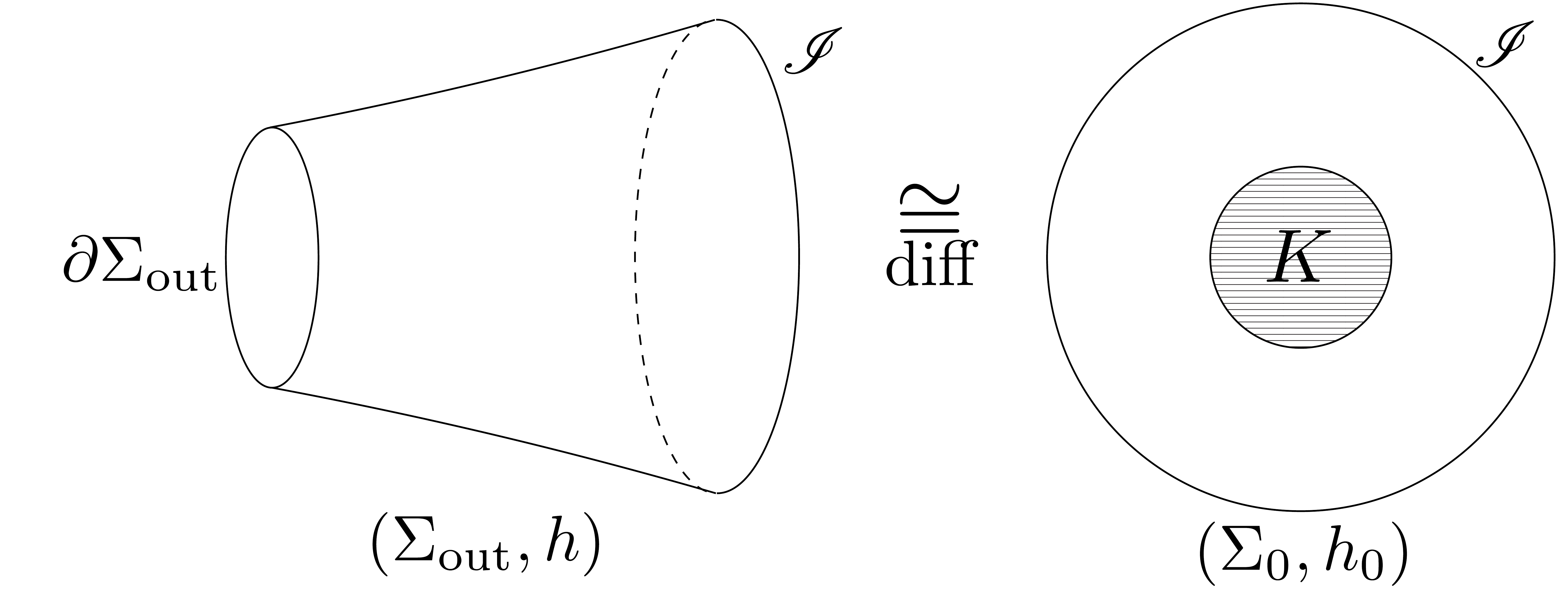}
\caption{Example of geometry appearing in Theorem~\ref{thm:BreCho}. On the left we see an
asymptotic ``arm'' $(\Sigma_{\rm out}, h)$ of an asymptotically hyperbolic manifold, where $\partial
\Sigma_{\rm out}$ is outermost minimal. The region $\Sigma_{\rm out}$ is
    diffeomorphic to a subset $\Sigma_{0}\setminus K$ of the hyperbolic plane $(\Sigma_0, h_0)$, and $h_0$ and $h$ are both metrics on $\Sigma_{\rm out}$ (via a judicious choice of diffeomorphism).}
\label{fig:volcomp}
\end{figure}
We can now state the result of \cite{BreCho14, Cho14}, which after a
rephrasing reads
\footnote{
    Note that while the theorem is stated with a strict inequality $m>0$ in Ref.~\cite{BreCho14}, Proposition 5.3 in
    Ref.~\cite{Cho14} includes the case $m=0$, where we should take $K$ empty so that $\partial \Sigma_0 = \varnothing$.}
\begin{thm}[\cite{BreCho14}, \cite{Cho14}]\label{thm:BreCho} 
    Let $(\Sigma_0, h_0)$ be hyperbolic 3-space and $(\Sigma_{m}, h_{m})$ the totally geodesic hypersurface 
    of Schwarzschild-AdS of mass $m$, restricted to one side of the bifurcation surface. 
    Consider a Riemannian metric $h$ on $\Sigma_{\rm out}= \Sigma_{0}\setminus K$ where $K$ is a compact set with smooth
    connected boundary. Assume that we can pick $K$ so that
    \begin{itemize}
        \item  $(\Sigma_{\rm out}, h)$ is asymptotically hyperbolic, meaning that 
            \begin{equation}
            \begin{aligned}
            |h-h_0|_{h_0} &= \mathcal{O}(r^{-2-\delta}), \quad \delta >0, \\
            \lim_{r \rightarrow \infty} |D_{0}(h-h_0)|_{h_0} &= 0,
            \end{aligned}
            \end{equation}
            where $D_0$ is the metric-compatible connection of $h_0$, and where $r$ is defined through
            the coordinates
            \begin{equation*}
            \begin{aligned}
            h_0 = \frac{ 1 }{ 1+r^2 }\dd r^2 + r^2 \dd \Omega^2.
            \end{aligned}
            \end{equation*}
        \item The Ricci scalar of $(\Sigma, h)$ satisfies $R[h] \geq -6/L^2$.
        \item $\partial \Sigma_{\rm out}$ is an outermost minimal surface with respect to $h$, and 
            \begin{equation*}
            \begin{aligned}
                \Area[\partial \Sigma_{\rm out}, h] \geq \Area[\partial \Sigma_m, h_m]
            \end{aligned}
            \end{equation*}
            for some $m \geq 0$.     \end{itemize}
            Then the renormalized volumes are finite and satisfy
            \begin{equation*}
            \begin{aligned}
                \vol_{\rm ren}(\Sigma_{\rm out}, h) \geq \vol_{\rm ren}(\Sigma_{m}, h_m),
            \end{aligned}
            \end{equation*}
            where \begin{equation}
      \vol_{\rm ren}(\Sigma_{\rm out}, h) = \lim_{i \rightarrow \infty}\left[\vol(\Sigma_{\rm out} \cap \Omega_i, h) -
      \vol(\Omega_i, h_0) \right],
\end{equation}
and $\{\Omega_i\}$ is an exhaustion of $\Sigma_0$ by compact sets, where asymptotic hyperbolicity 
ensures that the renormalized volume is well defined and finite. 
    Moreover, if and only if equality holds, then $h$ is isometric to $h_m$.
\end{thm}
See Fig.~\ref{fig:volcomp} for an illustration for quantities appearing in this theorem. 
Viewing $(\Sigma_{\rm out}, h)$ as a subset of a complete Riemannian manifold $(\Sigma, h)$, the theorem says that the volume of
$(\Sigma, h)$ contained outside the outermost minimal surface of each asymptotic region is greater than the volume of
a totally geodesic slice of Schwarzschild-AdS on one side of the bifurcation surface -- provided that the first two conditions
hold. Note that if we take $(\Sigma_{\rm out}, h)$ to be one half of a
totally geodesic slice of Schwarzschild-AdS, then we find that as expected the volume of the static slice of Schwarzschild-AdS$_{4}$ monotonically increases with mass. As a special case, of course, for any positive mass the volume is greater than twice the volume of a static slice of pure AdS.

Before applying the result to volumes of maximal volume slices in AAdS$_4$ spacetimes, we must understand how restrictive asymptotic hyperbolicity is. 
We relegate the details to Appendix~\ref{sec:asymhyp}, where we prove the following:
\begin{lem}\label{lem:ashyp}
    Let $(M, g)$ be an asymptotically globally AdS$_4$ spacetime. Then maximal volume static-anchored slices $\Sigma$ in $(M, g)$ are asymptotically hyperbolic.
\end{lem}

We now prove Theorem~\ref{thm:mainthm_d3}.
\begin{proof}
By Lemma~\ref{lem:ashyp}, a maximal volume slice $\Sigma$ is asymptotically
hyperbolic. Next, by the Gauss-Codazzi equation, the WCC and maximality ($K=0$),  the Ricci scalar of the
induced metric $h$ on $\Sigma$ is bounded from below:
\begin{equation}
\begin{aligned}
    R[h] = K^{\alpha \beta}K_{ \alpha \beta} + 2n^a n^b\left(R_{ab}[g]-\frac{ 1 }{ 2
    }g_{ab}R[g] \right) = K^{\alpha \beta}K_{ \alpha \beta}-\frac{ 6 }{ L^2 } + 2
    \mathcal{E} \geq -\frac{ 6 }{ L^2 }, \\
\end{aligned}
\end{equation}
where $n^a$ is a unit normal to $\Sigma$.
Thus, subsets $\Sigma_{\rm out}$ of $\Sigma$ that lies outside a connected outermost minimal surface will
satisfy the criteria for Theorem~\ref{thm:BreCho}.  

Let $\sigma_i$ for $i=1,
    \ldots, n$ be the outermost minimal surface for a given connected
component $\mathscr{I}_{i}$ of $\mathscr{I}$. Minimal surfaces homologous to $\mathscr{I}_i$ are guaranteed to exist by the results of \cite{AndMet07, AndEic10} when
$n\geq 2$.\footnote{The result of \cite{AndMet07, AndEic10} guarantees the existence of a marginally outer
    trapped surface (MOTS) in an initial dataset with an outer trapped inner boundary and an outer untrapped
outer boundary. If we take our initial dataset $(\Sigma, h, K_{ab})$ and turn it into the new dataset $(\Sigma,
h, K_{ab}=0)$, then MOTS in this second dataset are exactly the minimal surfaces of $(\Sigma, h)$,
and so the theorems of \cite{AndMet07, AndEic10} guarantee the existence of minimal surfaces when we
have at least two asymptotic regions. Note that the results of \cite{AndMet07, AndEic10} do not rely
on constraint equations or energy conditions, so setting $K_{ab}=0$ without altering the intrinsic
    geometry causes no problems.} Assume each $\sigma_i$ is connected.
For $n=1$ we have that $\sigma_1$ can be empty. If $\sigma_1$ is not empty, then by two applications of
    Theorem~\ref{thm:BreCho} we have (since for any value of Area$[\sigma_{i}]$, there is an $m$ such that Area$[\sigma_{i}]\geq \mathrm{Area}[\partial \Sigma_{m}]$)
\begin{equation}\label{eq:volbndi}
    \vol[\Sigma_{\mathrm{out,}i}] \geq \frac{1}{2}\vol[\Sigma_{\mathrm{Schw}[\sigma_i]}] \geq
    \vol[\Sigma_{\mathrm{AdS}_4}],
    \end{equation}
    where $\Sigma_{\mathrm{Schw}[\sigma_i]}$ is the totally geodesic hypersurface of Schwarzschild-AdS whose
    bifurcation surface has the same area as $\sigma_i$. If $\sigma_1$ is empty, then
    a single application of Theorem~\ref{thm:BreCho} directly gives $\vol[\Sigma_{\mathrm{out,}1}]
    \geq \vol[\Sigma_{\mathrm{AdS}_4}]$. For $n=1$ this completes our proof, since $V[\Sigma]\geq
    V[\Sigma_{\mathrm{out}, 1}]$, so assume now
    $n > 1$.

    Take now $\Sigma_{\mathrm{out,}i}$ so that they are closed as subsets of $\Sigma$, that is, $\sigma_i$ is included in $\Sigma_{\mathrm{out,}i}$. 
    If the interiors of $\Sigma_{\mathrm{out,}i}$ and $\Sigma_{\mathrm{out,}j}$ intersect for some $i\neq j$, then either $\sigma_i$
    is fully contained in the interior of $\Sigma_{\mathrm{out,}j}$, or the two surfaces $\sigma_i$ and $\sigma_j$ ``weave'' through each other. 
    The former is not possible, by outermost minimality, so consider the second case.
    Then $\sigma_i$ has a subset $U_i$ lying in $\Sigma_{\mathrm{out,}j}$, and $\sigma_j$ a subset
    $U_j$ in $\Sigma_{\mathrm{out,}i}$ such that $\hat{\sigma}_i = (\sigma_{i}-U_i)
    \cup U_j$ and $\hat{\sigma}_j = (\sigma_{j}-U_j) \cup U_i$ lie in $\Sigma_i$ and $\Sigma_j$ respectively -- see
    Fig.~\ref{fig:minoverlap}. Outermost minimality of $\sigma_i$ implies that
    \begin{equation}
    \begin{aligned}
        A[\sigma_i] \leq A[\hat{\sigma}_i] \qquad \Rightarrow \qquad A[U_i] \leq A[U_j].
    \end{aligned}
    \end{equation}
    But this implies $A[\hat{\sigma}_j] \leq A[\sigma_j]$ contradicting the outermost minimality of $\sigma_j$, and so
    we conclude that $\Sigma_i$ and $\Sigma_j$ only can intersect on a measure $0$ set, giving finally that
    \begin{equation}
    \begin{aligned}
        V[\Sigma] \geq \sum_{i=1}^{n} V[\Sigma_{\mathrm{out,}i}] = n V[\Sigma_{\mathrm{AdS}_4}].
    \end{aligned}
    \end{equation}
    \begin{figure}
    \centering
    \includegraphics[width=0.3\textwidth]{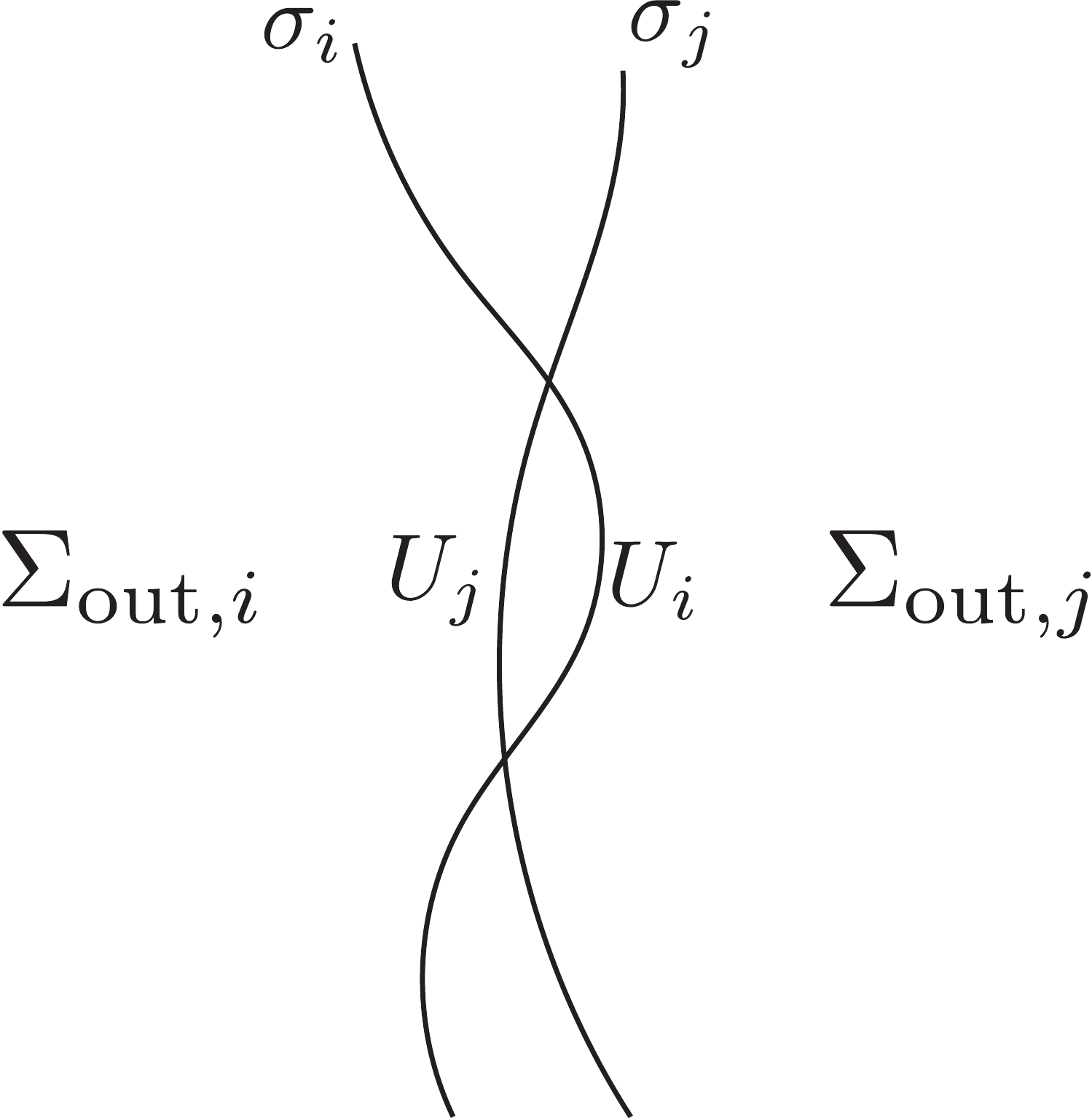}
    \caption{If two outermost minimal surface were to ``weave'' through each other, we would have a contradiction
        since either $(\sigma_{i}-U_i)
    \cup U_j$  would have area less than $\sigma_i$, or $(\sigma_{j}-U_j) \cup U_i$ would have area less than $\sigma_j$. }
    \label{fig:minoverlap}
    \end{figure}
\end{proof}
Could the above theorem be true also for higher dimensions?  A hint that this likely to be the case 
a comes from the following conjecture of Schoen \cite{Schoen}\footnote{
The conjecture is written in a different form in \cite{Schoen}. The version stated here can be found in~\cite{Bra97}.}
\begin{conj}[\cite{Schoen}]\label{conj:Schoen}
    Let $(\Sigma, h_0)$ be a closed hyperbolic Riemannian manifold with constant negative scalar curvature $R[h_0]$. Let $h$ be another metric on $\Sigma$ with scalar
    curvature $R[h] \geq R[h_0]$. Then $\vol(\Sigma, h) \geq \vol(\Sigma, h_0)$.
\end{conj}
\noindent As stated the conjecture applies to compact rather than conformally compact manifolds. Nevertheless, it
 adds plausibility to the possibility that Theorem~\ref{thm:mainthm_d3} holds in general dimensions. It also lends
 support to the idea that Theorem~\ref{thm:mainthm_d3} should hold without the assumption that the outermost minimal
 surface is connected, since the conformally compact generalization of Schoen's conjecture immediately implies
 Theorem~\ref{thm:mainthm_d3} with this assumption removed.
 In a moment we add further evidence for this by proving it for $d\geq 3$ spacetimes when we have maximal
 spatial symmetry and when $h_0$ is hyperbolic space.

 The proof of Theorem~\ref{thm:tfd_d3} follows that of Theorem~\ref{thm:mainthm_d3} \textit{mutatis mutandis} by replacing \eqref{eq:volbndi} with
\begin{equation}
\begin{aligned}
    \vol[\Sigma_{\mathrm{out,}i}] \geq \frac{1}{2}\vol[\Sigma_{\mathrm{Schw}[\sigma_i]}] \geq
    \frac{1}{2}\vol[\Sigma_{\mathrm{Schw}[\sigma^*]}],
\end{aligned}
\end{equation}
where now $\Sigma_{\mathrm{Schw}[\sigma_*]}$ is the totally
geodesic hypersurface of Schwarzschild-AdS with a bifurcation surface with area equal to the globally
minimal surface $\sigma_*$ on $\Sigma$. This second inequality holds due to monotonicity of the Schwarzschild-volume (and bifurcation surface area) with mass.

\subsubsection*{Proof of Theorems \ref{thm:mainthm_dall}, \ref{thm:tfd_dall} and \ref{thm:penrose}}\label{sec:withsym}
The proof of Theorem~\ref{thm:mainthm_dall} is threefold: first, we establish that the intrinsic metric of an
extremal hypersurface in an asymptotically AdS spacetime with maximal spatial symmetry can be written as
\begin{equation}
\begin{aligned}
    \dd s^2 &= \frac{ 1 }{ k + \frac{ r^2 }{L^2 } - \frac{ \omega(r) }{ r^{d-2} } }\dd r^2 + r^2 \dd
           Q_{k}^2,
\end{aligned}
\end{equation}
where $\omega(r)$ is a monotonously non-decreasing function whenever the WCC holds. 
Second, we show that at a stationary surface $r=r_*$, $\omega$ satisfies the equation
\begin{equation}\label{eq:omegstar}
\begin{aligned}
k + \frac{ r_*^2 }{L^2 } - \frac{ \omega(r_*) }{ r_*^{d-2} } = 0,
\end{aligned}
\end{equation}
and by the monotonicity of $\omega$ the volume contained outside $r_*$ is greater than the one we obtain by replacing $\omega(r)$ with the constant $\omega(r_*)$. Third, we note that constant $\omega$ corresponds to a constant$-t$ slice of
Stat$_k$, and taken together with \eqref{eq:omegstar} we see that the geometry outside $r_*$ after
the replacement $\omega(r)\rightarrow \omega(r_*)$ is exactly that of half a constant-$t$ slice of
Stat$_k$. Thus, the volume outside a stationary surface at $r_*$ is greater than that of half a constant-$t$ slice of
Stat$_k$ with mass parameter $\omega(r_*)$, and this is in turn greater than the volume of a static
AdS slice. 

We show in Appendix~\ref{sec:GerHawMass} that there is a geometric functional $\omega[\sigma,
\Sigma]$ that, in an
AAdS spacetime with maximal spatial symmetry, reduces to
$\omega(r)$ when $\sigma$ is a surface invariant under the spatial symmetry -- which we shall call a `symmetric' surface -- having area radius $r$
-- henceforth denoted $\sigma_r$. The quantity
\begin{equation}
    \begin{aligned}
        M_{\Sigma}[\sigma_r] & \equiv \omega[\sigma_r, \Sigma] \frac{(d-1) \Omega_{k}}{16 \pi G_{N}}
\end{aligned}
\end{equation}
asymptotes to the spacetime mass as $\sigma_r$ approaches the asymptotic
boundary \cite{HerHor04}, where we remind that $\Omega_k$
is the volume of the unit sphere, the plane or the unit hyperbolic plane.\footnote{
    For the
    falloffs and symmetries assumed in this paper this mass agrees with the CFT
    energy up to a fixed offset corresponding to the Casimir energy. For weaker falloffs we have
    that $M_{\Sigma}[\sigma]$ diverges towards the boundary. For $d=3$, $M_{\Sigma}[\sigma]$ is the
    so-called Geroch-Hawking mass \cite{Haw68, Ger73,JanWal77, ChrSim01, Wan01,BraHay06}.}  In Sec.~\ref{sec:lloyd} we will see that $\omega(r)\geq 0$ whenever the WCC
holds, $d\geq 3$ and the spatial symmetry is spherical or planar. As an aside, in the asymptotically flat context, this
quantity was used early partial proofs of the positive energy
theorem~\cite{Ger73,JanWal77}. 

\begin{figure}
\centering
\includegraphics[width=0.7\textwidth]{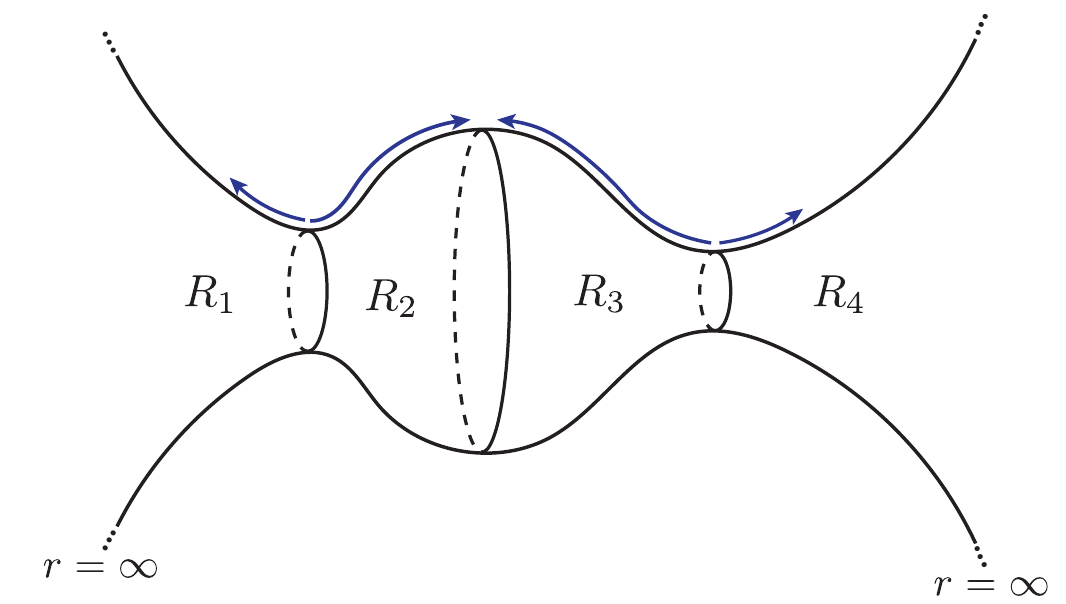}
\caption{Example of a two-dimensional spherically symmetric hypersurface $\Sigma$ with three stationary
surfaces. This manifold can be covered by four coordinate systems of the form
\eqref{eq:omegametric}, and $\omega(r)$ increases along the blue arrows. }
\label{fig:omega}
\end{figure}
We now prove the properties of $\omega(r)$ claimed in the beginning of this subsection (see
Fig.~\ref{fig:omega} for a illustration of the quantities appearing in this lemma).
\begin{lem}\label{lem:flow} 
    If $\Sigma$ is a complete maximal volume slice in an AAdS$_{d+1}$ spacetime with maximal spatial symmetry satisfying
    the WCC,  then 
    \begin{itemize}
    \item Any region $R$ that lies on $\Sigma$ between two stationary symmetric surfaces can be covered by the coordinates
        \begin{equation}\label{eq:omegametric}
        \begin{aligned}
            \dd s^2|_{\Sigma} = \frac{ 1 }{ k + \frac{ r^2 }{L^2 } - \frac{ \omega(r) }{ r^{d-2} } }\dd r^2 + r^2 \dd
            Q_{k}^2.
        \end{aligned}
        \end{equation}
        In particular, this is true when one of the two stationary surfaces is $\varnothing$ ($r=0$) or asymptotic infinity ($r=\infty$). 
    \item 
    \begin{equation}
    \begin{aligned}
        \omega'(r) \geq 0.
    \end{aligned}
    \end{equation}
    \item There is always an outermost symmetric stationary surface $\sigma_{\tilde{r}}$ such that \eqref{eq:omegametric}
        covers all $r > \tilde{r}$, where $\sigma_{\tilde{r}} = \varnothing$ and $\tilde{r}=0$ is allowed.
    \end{itemize}
\end{lem}
\begin{proof}
Since by assumption $\Sigma$ is foliated by symmetric surfaces, locally on $\Sigma$ we can pick an ADM coordinate system
with respect to this foliation with vanishing shift and unit lapse:
\begin{equation*}
\begin{aligned}
    \dd s^2 = \dd \rho^2 + \gamma_{ij}(\rho, x) \dd x^i \dd x^{j},
\end{aligned}
\end{equation*}
where $x^i$ are coordinates on the symmetric surfaces. Maximal spatial symmetry implies that
$\gamma_{ij}(\rho, x) \dd x^i \dd x^{j}=f(\rho)\dd Q_k^2$ for some function $f(\rho)$ that is nonnegative by the
spacelike signature of $\Sigma$. Defining the new
coordinate $r$ through
\begin{equation}
\begin{aligned}
    r = \sqrt{f(\rho)},
\end{aligned}
\end{equation}
gives  a metric of the form
\begin{equation}\label{eq:cancoords}
\begin{aligned}
    \dd s^2 = B(r) \dd r^2 + r^2 \dd Q_k^2, \qquad B(r)
    = \frac{ 4 f }{ (\partial_\rho f)^2 }.
\end{aligned}
\end{equation}
These coordinates break down at $r=0$ and at radii where $\partial_\rho f(\rho)=0, f\neq 0
\Leftrightarrow B=\infty$, i.e.\ at
    symmetric surfaces that are stationary. The case of $B=0$ corresponds to a singularity (the Kretschmann scalar
    is divergent for $d\geq 3$, and there is a conical singularity for $d=2$), 
which is incompatible with completeness of $\Sigma$.\footnote{A discontinuous but nonzero $B(r)$ 
corresponds to a distributional energy shock on $\Sigma$, which does not lead to
breakdown of the coordinates.} To bring the metric to
the required form, we simple define $\omega(r)$ through
\begin{equation}
\begin{aligned}
    B(r) = \frac{ 1 }{ k + \frac{ r^2 }{ L^2  }  - \frac{ \omega(r) }{ r^{d-2} }}.
\end{aligned}
\end{equation}

For the second point, we use the Gauss-Codazzi equation applied to $\Sigma$, which yields
    \begin{equation}\label{eq:omegamaster}
\begin{aligned}
    (d-1) \frac{ \omega'(r) }{ r^{d-1} } = K^{\alpha\beta}K_{\alpha\beta} + 2 \mathcal{E}\geq 0,
\end{aligned}
\end{equation}
where we have used the WCC, extremality of $\Sigma$ and positivity of $K_{\alpha\beta}K^{\alpha\beta}$. 

    Finally, since the spacetime has finite mass, $\omega(r)$ converges to some finite number $m$
at large $r$. Thus, the quantity
\begin{equation}
\begin{aligned}
    B(r)^{-1} = k + \frac{ r^2 }{ L^2 } - \frac{ \omega(r) }{ r^{d-2} }
\end{aligned}
\end{equation}
must have a last zero $\tilde{r}$, unless it has none, in which case we set $\tilde{r}=0$.
Consequently, for $r > \tilde{r}$ there are no stationary surfaces, and so a single coordinate patch \eqref{eq:cancoords} covers this region.
\end{proof}

Next we prove our main volume comparison result:
\begin{thm}\label{thm:mainComparison}
    Let $(M, g)$ be a maximally spatially symmetric asymptotically AdS$_{d+1\geq 3}$ spacetime satisfying the
    WCC, and let $\Sigma$ be a static-anchored complete maximal volume slice of $M$. Let
    $\sigma_{r_i}$ be the outermost symmetric stationary surface with respect to the connected
    component $\mathscr{I}_i$ of $\mathscr{I}$. Then 
    \begin{equation}
        \begin{aligned}\label{eq:Vineq}
        V(\Sigma)  \geq \sum\limits_{i} V(\Sigma_{m_{i}}),
    \end{aligned}
    \end{equation}
    with
    \begin{equation}
    \begin{aligned}
        V(\Sigma_{m_i}) = \Omega_{k}\int_{r_i}^{\infty} \dd r \frac{ r^{d-1} }{ \sqrt{k + \frac{ r^2 }{ L^2 } - \frac{
            m_i }{
            r^{d-2} }}},
    \end{aligned}
    \end{equation}
    and with $m_i$ defined by
    \begin{equation}
    \begin{aligned}
        k + \frac{ r_i^2 }{ L^2 } - \frac{ m_i }{ r_i^{d-2} } = 0.
    \end{aligned}
    \end{equation}
    If $\mathscr{I}$ is connected, \eqref{eq:Vineq} is an equality if and only if $\Sigma$ 
    has the same extrinsic geometry as that of
    a static slice of pure AdS. If $\mathscr{I}$ has multiple connected components, \eqref{eq:Vineq} is an equality if and only if $\Sigma$ has
    the same extrinsic geometry as that of a totally geodesic slice of a static AdS black hole with spherical,
    planar, or hyperbolic horizon topology.
\end{thm}
\begin{proof}
It is illustrative to begin with the one-sided case. Let $\sigma_{r_1}$ be the outermost symmetric stationary surface on $\Sigma$.
If there is none then $r_{1}=0$. Clearly the full volume of $\Sigma$ is greater or equal to the
    volume outside $\sigma_{r_1}$, and since by Lemma~\ref{lem:flow}, $\omega'(r)\geq 0$, and the
    region outside $\sigma_{r_1}$ is covered by one patch of the coordinates \eqref{eq:omegametric}, we get 
    \begin{equation}\label{eq:Ik}
\begin{aligned}
    V(\Sigma) &\geq \Omega_{k}\int_{r_{1}}^{\infty} \frac{ \dd r r^{d-1} }{ \sqrt{k + \frac{ r^2 }{ L^2  } - \frac{ \omega(r) }{
        r^{d-2} }} } \\
              &\geq \Omega_{k}\int_{r_{1}}^{\infty} \frac{ \dd r r^{d-1} }{ \sqrt{k + \frac{ r^2 }{ L^2  } - \frac{
                  \omega(r_1) }{ r^{d-2} }} } \\
              &= V(\Sigma_{\omega(r_1)}),
\end{aligned}
\end{equation}
    where $\omega(r_1)$ satisfies $0= k + \frac{ r_1^2 }{ L^2  } - \frac{\omega(r_1) }{ r_1^{d-2} }$ by stationarity of
    $\sigma_{r_1}$. The two-sided cases follows from the same argument applied twice.

    Finally,  the above inequalities can all be replaced by equalities if and only if $\omega(r)$ is
    constant, and with $n \in \{1, 2\}$
    coordinate patches fully covering $\Sigma$ when there are $n$ conformal boundaries.
    This implies the intrinsic geometry corresponds to a constant$-t$ slice of
    either pure AdS ($n=1$) or Stat$_k$ ($n=2$). From \eqref{eq:omegamaster},
    we see that $\omega$ is constant  only if $K_{ab}=0$, and so the full extrinsic geometry is
    that of a constant$-t$ slice embedded in pure AdS or Stat$_k$. 
\end{proof}

The final ingredient for the proof of Theorems \ref{thm:mainthm_dall} and \ref{thm:tfd_dall} is the 
behavior of the volume of totally geodesic slices of Stat$_k$ as a function of the mass. 
In the case $d=2$, i.e. for the BTZ black hole, we readily find that the volume of a 
totally geodesic hypersurface relative to that of static slices of pure AdS is a positive constant
independent of mass. Consider thus now $d\geq 3$, and define the vacuum-subtracted volumes: 
\begin{equation}
\begin{aligned}
    I_{k}(m) &= \Omega_{k} \lim_{r_c \rightarrow \infty}\left[\int_{r_h(m)}^{r_c}\dd r \frac{ r^{d-1} }{ \sqrt{k + \frac{ r^2
    }{ L^2 }  - \frac{ m }{ r^{d-2} }} } - \int_{0}^{r_c}\dd r \frac{ r^{d-1} }{ \sqrt{k + \frac{ r^2
    }{ L^2 }} } \right], \qquad k\in \{0, 1\} \\
    I_{-1}(m) &= \Omega_{-1} \lim_{ r_c \rightarrow \infty } \left[ \int_{r_h(m)}^{r_c}\dd r \frac{ r^{d-1} }{ \sqrt{-1 +
    \frac{ r^2 }{ L^2 } - \frac{ m }{ r^{d-2} }} } -  \int_{L}^{r_c}\dd r \frac{ r^{d-1} }{ \sqrt{-1 +
    \frac{ r^2 }{ L^2 }  } }\right].
\end{aligned}
\end{equation}
Here $r_h(m)$ satisfies $k + \frac{ r_h^2 }{ L^2 } - \frac{ m }{ r_h^{d-2} }=0$. 
For $k\in\{0, 1\}$, smoothness and completeness of $\Sigma$ is only compatible with $m \geq 0$. For $k=-1$, we can also have black
holes with $\hat{m}\leq m<0$, where \cite{Bir98, Emp99}
\begin{equation}
    \begin{aligned}
        \hat{m}= - \frac{ 2 L^{d-2} }{ d-2 }\left(\frac{ d-2 }{ d }\right)^{\frac{ d }{ 2 }}.
\end{aligned}
\end{equation}

The function $I_{0}(m)$ was computed explicitly in \cite{ChaMar16}
was be seen to be positive and monotonically increasing with $m$ \cite{ChaMar16}. For
$k=1$ and $d=3$ the same properties follow from
Theorem \ref{thm:BreCho}, as discussed previously. In the Appendix~\ref{app:Ik} we prove that this
extends to all $d\geq 3$ by generalizing a proof from \cite{BreCho14}:
\begin{lem}
    The volume of a totally geodesic hypersurface of Schwarzschild-AdS$_{d+1}$ is monotonically increasing
    with mass: 
    \begin{equation*}
    \begin{aligned}
    \partial_m I_{1}(m) > 0.
    \end{aligned}
    \end{equation*}
\end{lem}
For $k=-1$ numerical studies \cite{ChaMar16} indicate that
$I_{-1}(m)$ is non-negative, has a global minimum equal to zero at $m=0$, and increases
monotonically with $m$ for $m>0$; at $m<0$ it appears to monotonically decrease, and it diverges at $m = \hat{m}$.
However, at the moment there is no rigorous proof these results, so we only state the $k\in\{0,
1\}$ cases as theorems.

It follows from the above that that totally geodesic hypersurfaces of Stat$_k$ black holes for
$k\in\{0, 1\}$ have  volumes greater than $2 \vol[\Sigma_{\mathrm{AdS}_{d+1}}]$, 
which implies that the right-hand side of \eqref{eq:Vineq} is
greater than one or two times $\vol[\Sigma_{\mathrm{AdS}_{d+1}}]$. This proves
Theorem~\ref{thm:mainthm_dall}.  

To prove Theorem \ref{thm:tfd_dall}, note that  $\partial_m I_k(m) > 0$ for $k\in
\{0, 1\}$ implies $\partial_A I_k(m(A)) > 0$, with $A$ the area of the bifurcation surface,
since the area of the bifurcation surface increases with mass. In the case of two conformal boundaries, we
can now bound the right-hand side of
\eqref{eq:Vineq} from below by $\vol(\Sigma_*)$, where $\Sigma_*$ is a totally geodesic slice of
Stat$_k$ with bifurcation surface area equal to the globally minimal surface area.
 
Finally, we note that the stationarity condition $k+r_*^2 - \omega(r_*)/r_*^{d-2}=0$ together
with the monotonicity of $\omega$ (giving $\omega(\infty)\geq \omega(r_*)$) immediately proves Theorem~\ref{thm:penrose}.

\section{Upper Bounds on Complexity Growth}\label{sec:lloyd}
\subsection{Lloyd's bound}
Lloyd's bound~\cite{Llo00} is a conjectured constraint on general quantum systems given by
\begin{equation}
\begin{aligned}
    \frac{ 1 }{ \Delta \tau } \leq \frac{ 2E }{ \pi },
\end{aligned}
\end{equation}
where $E$ is the energy of the system and $\Delta \tau$ the minimal time that a quantum
state takes to transition from an initial state to a new orthogonal state. This relation constitutes a
bound on the maximal speed of computation. The existence of a bound on the growth of either bulk
volume or action in terms of CFT energy would provide evidence for the holographic complexity proposals, and  there is a large
literature devoted to investigations of holographic complexity growth \cite{StaSus14, SusZha14, RobSta14, BroRob15,
BroRob15b, CaiRua16, MomFai16, PanHua16, CaiSas17, WanYan17, GuoWei17, CarCha17, CotMon17, HosSey17, SwiWan17, Moo17b, AnPen18, AliFar18, CanHen18, ChaMar18, ChaMar18b, AnCai18, AuzBai18, Nag18, CouEcc18,
 Gho18, MahRoy18, TanVaz18, Jia18, Men18, FanGuo18, FenLiu18, Jia18b, AliBab19, Age19, Jia19, Cag19, FanLia19, AnCai19,
ChaChe19,Nag19, ZhoKua19, BhaDas20, PanLi20, Li20, Raz20, BalHen20, ZhoKua21}. In particular, the
points of interest is whether this growth is linear at late times,\footnote{But not so late as to
reach the quantum recurrence time \cite{Sus14a}, where complexity saturates and fluctuates. That circuit complexity indeed
grows linearly until the recurrence time was recently argued for random quantum circuits in~\cite{HafFai21}.} as hypothesized in~\cite{Sus14a}
, and whether there is
an upper bound on the growth of holographic complexity of the form 
\begin{equation}
    \begin{aligned}\label{eq:Lloyd}
\frac{ \dd \mathcal{C} }{ \dd \tau  } \leq c E,
\end{aligned}
\end{equation}
for the CFT energy $E$ and some constant $c$. In the holography literature \eqref{eq:Lloyd} is typically
referred to as Lloyd's bound. For the CV-proposal with the AdS radius as
$L$, it is natural to set $c=\frac{ 8\pi }{ d-1 }$, so
that~\eqref{eq:Lloyd} is an upper bound in terms of the late time volume growth in Schwarzschild-AdS as $M\rightarrow
\infty$ (with the identification of $E$ with
the Schwarzschild mass $M$ \cite{CarCha17}). For charged and rotating spacetimes it is expected that
\eqref{eq:Lloyd} can be strengthed via the inclusion of terms depending on the charge and angular
momentum on the right hand side of \eqref{eq:Lloyd}~\cite{BroRob15, BroRob15b, CaiRua16}.

While specific case studies are invaluable as tests of the complexity proposals, a robust test of whether
\eqref{eq:Lloyd} holds for holographic complexity via CV requires a broad investigation of the growth of
$\mathcal{C}_V$, which has heretofore been elusive.\footnote{Some partial results on Lloyd's bound at late times
in static spacetimes exist for the Complexity=Action proposal \cite{Yan20}.} Here we make strides towards closing that lacuna:

\begin{thm}\label{thm:Lloyd}
    A WCC-satisfying spherically symmetric AAdS$_{d+1\geq 4}$ spacetime with any number of
     minimally coupled $U(1)$ gauge
    fields and scalars, both charged and neutral satisfies
    \begin{equation}\label{eq:provenLloyd}
        \begin{aligned}
            \Big|\frac{\dd \CV}{\dd \tau}\Big| \leq \frac{ 8\pi M }{ d-1 }f(M),
    \end{aligned}
    \end{equation}
    where $M$ is the spacetime mass and 
    \begin{equation}\label{eq:fexplicit}
    \begin{aligned}
    f = \begin{cases}
        1 & M \leq \hat{M} \equiv \frac{ (d-1)\vol(S^{d-1})L^{d-2} }{ 16 \pi G_{N}}\left[\frac{ d }{  4(d-2) }\right]^{\frac{ d-2 }{ 2 }} \\
        1 + \left(\frac{ M }{ \hat{M} }\right)^{\frac{ 1 }{ d-2 }} & \mathrm{otherwise}.
    \end{cases}.
    \end{aligned}
    \end{equation}
   Here $\tau$ is the time parameter in the boundary conformal frame $-\dd \tau^2 + L^2 \dd \Omega^2$. 
\end{thm}

\noindent Theorem~\ref{thm:Lloyd} is an immediate proof that $    \frac{ \CV(\tau) }{ \tau }$ is bounded
at late times (in spacetimes subject to our assumptions). Thus if $\CV$ is not oscillating at late times, then the late
time growth of $\CV$ is at most
linear. Finally, if in some multiboundary spacetime the different boundaries have different masses, then $M$ in \eqref{eq:provenLloyd} is the one
associated to the conformal boundary at which we perturb the time anchoring.

Provided that we can identify the spacetime mass $M$ with the CFT energy $E$, our bound is identical to
Lloyd's bound \eqref{eq:Lloyd} for small masses, while for large masses ($M\geq \hat{M}$) it is qualitatively
similar, although the energy dependence is now non-linear.
Thus, for spacetimes with ordinary falloffs,
Theorem~\ref{thm:Lloyd} constitutes a partial proof of a complexity growth bound, albeit with non-linear energy
dependence, and where it is the vacuum-subtracted CFT energy that appears in the bound.\footnote{Our assumptions about the falloffs are naturally violated in spacetimes where the ordinary gravitational mass $M$, as defined either as in 
\cite{BalKra99b} (with subtraction of Casimir energy) or as the Geroch-Hawking mass at infinity, diverges. In such case the growth of
extremal surface volumes can be infinite, as in e.g. \cite{Moo17b}. In such cases
the field theory energy after subtraction of the Casimir energy is no longer equal to $M$, and
additional counterterms depending on scalar fields contributes to give a finite CFT energy $E$ \cite{HenMar02, Hor03,
HenMar04, HerHor04, HerMae04, HenMar06}.
In this scenario \eqref{eq:provenLloyd} is true but vacuous since both sides of the inequality diverge.} 

Before turning to proving our result, let us finally note that the proof of \eqref{eq:provenLloyd} puts the
CV and CA proposals on quite different footings when it comes to complexity growth. 
It is known that Lloyd's bound cannot hold for the CA proposal \cite{CarCha17, BalHen20}; spacetimes with ordinary
AdS asymptotics and finite gravitational mass can have negative and divergent action growth \cite{CarCha17}.

\paragraph{An easy way to compute complexity change:}\label{sec:easycdot}
Under the assumption of maximal spatial symmetry, the extrinsic geometry of the initial maximal volume slice at some
time is sufficient to yield a simple expression for $\dCV$: given a one-parameter family of extremal slices $\Sigma_{\tau}$
($\Sigma_{0}\equiv \Sigma$) with boundary $\partial \Sigma_{\tau}$, extremality ensures that any volume change from a
perturbation must be a boundary term (see for example the appendices of \cite{FisWis16,BaoCao19}):
\begin{equation}
\begin{aligned}
    \frac{ \dd V[\Sigma_{\tau}] }{ \dd \tau } = \int_{\partial \Sigma_{\tau}} \eta^a N_a, 
\end{aligned}
\end{equation}
where $\eta^a = (\partial_{\tau})^a$ is the displacement field for $\partial \Sigma_{\tau}$ and $N_a$ the outwards unit normal
to $\partial\Sigma_{\tau}$ in $\Sigma_{\tau}$. We want to work with $\Sigma_{\tau}$ which technically do not have
boundaries. However, the above formula can be used by cutting off $\Sigma_{\tau}$ at a finite radius, applying the
formula, and then afterward taking the limit where $\partial \Sigma_{\tau}$ goes to the conformal boundary. 

Assume now that the gravitational mass is finite and remains so on evolution of the initial data on $\Sigma$. 
Expressing $\eta^a N_a$ in terms of quantities over which we have control, carried out in Appendix~\ref{app:Cdot}, we find for
spherical or planar symmetry that
\begin{equation}
\begin{aligned}
    \frac{ \dd \CV }{ \dd \tau } = \frac{ \Omega_{k} }{ G_N L (d-1) }\lim_{r \rightarrow
    \infty}r^{d}K_{ab}r^{a}r^{b},
\end{aligned}
\end{equation}
where $r^{a}$ is the outwards unit normal to $\sigma_r$ tangent to $\Sigma$. Thus, knowledge of
$K_{ab}$ near the
boundary is sufficient to compute $\dCV$.\footnote{Note however that a near boundary analysis is not enough
to obtain $K_{ab}$ in the first place. If we treat the determination of an extremal slice as a shooting
problem from the conformal boundary, a general initial condition at the boundary is incompatible with smoothness in the
bulk, so determining $\dCV$ still requires deep-bulk information.} 
In Appendix~\ref{app:Cdot} we show that we can rewrite this expression in Einstein+matter gravity as
\begin{equation}
    \begin{aligned}\label{eq:trappedCdot}
    \frac{ \dd \CV }{ \dd \tau }
    &=  \frac{ 1 }{ G_N L (d-1)  }\left[r_{\rm throat} \int_{\partial \Sigma_{\rm out}}\sqrt{2\theta_{k}\theta_{\ell}} + 8\pi G_N
    \int_{\Sigma_{\rm out}} r T_{ab}n^{a}r^{b} \right], \\
\end{aligned}
\end{equation}
where now $\partial \Sigma_{\rm out}$ is an outermost stationary surface on $\Sigma$ at radius $r=r_{\rm throat}$, and
$\Sigma_{\rm out}$ is the part of $\Sigma$ outside this surface. Here $n^a$ is the future normal to $\Sigma$, $k$ and $\ell$ the future null normals to
$\partial \Sigma_{\rm out}$ normalized so $k\cdot \ell = -1$, and $\theta_{k}$ and $\theta_{\ell}$ are
the null geodesic expansions.  It follows from \eqref{eq:Htheta} in the
appendix that $\theta_{k}\theta_{\ell}\geq 0$ by stationarity of $\partial \Sigma_{\rm out}$, so the square root is always real
This expression \eqref{eq:trappedCdot} is a 
special case of the more general momentum-complexity correspondence proven in \cite{BarMar20}, 
with the difference that we write our integral on only part of $\Sigma$ at the cost of a boundary term.
In spacetimes where $n^a r^b T_{ab}=0$, which includes AdS-Schwarzschild and AdS-Reissner-Nordstr{\"o}m, we see that the
complexity growth is proportional $\sqrt{\theta_{k}\theta_{\ell}}$ at the throat of $\Sigma$. This suggests a connection between trapped and anti-trapped regions and complexity growth, at least in these spacetimes. The more trapped the throat of
$\Sigma$ is, the faster the complexity growth.

With our convenient formulae for $\dCV$ in hand, we can turn to proving that the mass indeed bounds
$\dCV$ in Einstein-Maxwell-Scalar theory.

\subsection*{Proof of a Lloyd's bound in Einstein-Maxwell-Scalar theory (with spherical
symmetry)}\label{sec:lloydproof}
\subsubsection*{Explicit constraint solutions}
Consider Einstein-gravity minimally coupled to a $U(1)$ gauge field and a charged scalar:
\begin{equation}
    \begin{aligned}\label{eq:Scharged}
        S &= \frac{ 1 }{ 8\pi G_N }\int_{M}\dd^{d+1}x \sqrt{-g} \left[\frac{ 1 }{ 2 }R + \frac{
            d(d-1) }{ 2L^2 } - |\mathcal{D} \phi|^2  -
        V(\phi, \phi^{\dag}) -\frac{ 1 }{ 4 }F_{ab}F^{ab}\right],
\end{aligned}
\end{equation}
where $\mathcal{D}$ is covariant derivative associated with the $U(1)$ gauge field.
The stress tensor reads
\begin{equation}
\begin{aligned}
    8\pi G_N T_{ab} = \mathcal{D}_{a}\phi (\mathcal{D}_{b}\phi)^{\dag} +\mathcal{D}_{b}\phi (\mathcal{D}_{a}\phi)^{\dag}
    - g_{ab}|\mathcal{D}\phi|^2 - g_{ab}V(\phi,
    \phi^{\dag}) +
    F\indices{_{a}^{c}} F_{bc} - \frac{ 1 }{ 4 }g_{ab} F_{cd} F^{cd}.
\end{aligned}
\end{equation}

We now restrict to spherical, planar or hyperbolic symmetry and work in the coordinates \eqref{eq:cancoords}:
\begin{equation}
\begin{aligned}
    \dd s^2|_{\Sigma} = B(r) \dd r^2 + r^2 \dd \Omega^2_k = \frac{ 1 }{ k + \frac{ r^2 }{ L^2 } - \frac{ \omega(r) }{
        r^{d-2} } }\dd r^2 + r^2 \dd \Omega^2_k.
\end{aligned}
\end{equation}
We then get that the energy density and radial energy current on $\Sigma$ are
\begin{equation}
    \begin{aligned}\label{eq:JEexplicit}
    \mathcal{E} &= |\mathcal{D}_t \phi|^2 + \frac{ 1 }{ B(r) }|\mathcal{D}_{r}\phi|^2 + V(\phi, \phi^{\dag}) + \frac{ 1 }{
        2 B(r) }F_{tr}^2 + \frac{ 1 }{ 4 }F^{ij}F_{ij}\\
    J_{r} &= \mathcal{D}_t \phi (\mathcal{D}_{r}\phi)^{\dag} + \mathcal{D}_r \phi (\mathcal{D}_t\phi)^{\dag},
\end{aligned}
\end{equation}
where $\mathcal{D}_t \equiv n^a \mathcal{D}_a$ and $F_{tb}\equiv n^a F_{ab}$, with $n^a$ is the future unit normal to $\Sigma$.
The $i, j$-indices run over coordinates of the transverse space. Symmetry sets $F_{ri}=F_{ti}=0$. Certain nonzero
$F_{ij}$ are possible, but the constraints of Maxwell theory might set these to zero, depending on the topology of
$\Sigma$. The WCC, here equivalent to the WEC, holds when $V(\phi, \phi^{\dag})\geq 0$.

Symmetry and $K=0$ dictates that
\begin{equation}\label{eq:Ksym}
\begin{aligned}
    K_{\alpha\beta}\dd x^{\alpha}\dd x^{\beta}= K_{rr}(r) \left[ \dd r^2 - \frac{ r^2 }{ B(r) (d-1) }
    \dd\Omega^2_{k} \right].
\end{aligned}
\end{equation}
Defining the function $\mathcal{K}$ through $K_{rr} = B(r) \mathcal{K}(r)$, the constraint equations
simply read\footnote{Note that \eqref{eq:constraintEqsE} and \eqref{eq:constraintEqsJ} are true independent of the particular matter we are studying.}
\begin{align}
        (d-1) \frac{ \omega'(r) }{ r^{d-1} } &= 2 \mathcal{E}(r) + \frac{ d }{ d-1
        }\mathcal{K}(r)^{2}, \label{eq:constraintEqsE}\\
        \frac{ 1 }{ r^{d} }\frac{ \dd }{ \dd r }\left[ r^{d}\mathcal{K}(r)\right] &= J_r(r). \label{eq:constraintEqsJ}
\end{align}
These equations are linear ODEs and are readily solved to give
\begin{equation}
    \begin{aligned}\label{eq:scalarSols}
        \mathcal{K}(r) &= \frac{ 1 }{ r^{d} }\left\{ \mathcal{K}(r_0)r_0^{d} + \int_{r_0}^{r}\dd \rho \rho^{d}\left[ \mathcal{D}_t\phi
    (\mathcal{D}_{r}\phi)^{\dag},
    + \mathcal{D}_r \phi (\mathcal{D}_t \phi)^{\dag} \right] \right\} \\
        \omega(r) &= e^{-h(r)}\left[ \omega(r_0) + \frac{ 1 }{ d-1 }\int_{r_0}^{r}\dd \rho e^{h(\rho)} \rho^{d-1} 
        \chi(\rho) \right], \\ 
        h(r) &= \frac{ 1 }{ d-1 }\int_{r_0}^{r} \dd \rho \rho \left( 2 |\mathcal{D}_r \phi|^2 + F_{tr}^2 \right), \\ 
        \chi(r) &=  \frac{ d }{ d-1 }\mathcal{K}(\rho)^2 + 2\left(k +  \frac{ \rho^2 }{ L^2 } \right)|\mathcal{D}_r\phi|^2 +
        2|\mathcal{D}_t\phi|^2 + 2
        V\left(\phi\right) \\ 
        & \qquad + \left(k + \frac{ \rho^2 }{ L^2 }\right)F_{tr}^2 + \frac{ 1 }{ 2 }F^{ij}F_{ij}.  \\
\end{aligned}
\end{equation}
The solution for a canonically normalized real scalar is obtained by sending $\phi \rightarrow \frac{ 1 }{ \sqrt{2} }\phi$ and
replacing $\mathcal{D}_{a}$ with $\partial_{a}$. 

Noting now that $K_{ab}r^{a}r^{b}=\mathcal{K}$ in an outermost coordinate patch where $r^{\alpha}$
points to increasing $r$, we find that the complexity change for spherical or planar symmetry is
\begin{equation}
    \begin{aligned}\label{eq:cdotk}
    \frac{ \dd \CV }{ \dd \tau } &= \frac{ \Omega_{k} }{ G_N L (d-1) }\lim_{r \rightarrow \infty}
    r^{d}\mathcal{K}(r),
\end{aligned}
\end{equation}
while the mass reads
\begin{equation}\label{eq:omegaM}
\begin{aligned}
    M = \lim_{r \rightarrow \infty} M_{\Sigma}[\sigma_r] = \frac{ (d-1)\Omega_{k} }{ 16\pi G_N
    }\omega(\infty).
\end{aligned}
\end{equation}
Thus, proving Lloyd's bound now amounts to constraining the asymptotic values of the solution \eqref{eq:scalarSols}.

\subsubsection*{Some technical lemmas}
We begin by establishing positivity of $\omega$:
\begin{lem}\label{lem:minorlem}
	$\omega(r)$ is positive on any complete maximal volume slice in an AAdS$_{d+1\geq 4}$ spacetime satisfying the WCC with spherical or planar symmetry. 
\end{lem}
\begin{proof}
Consider first a one-sided spacetime. 
A finite Ricci and Kretschmann scalar as $r \rightarrow 0 $ on $\Sigma$ requires $\lim_{r \rightarrow 0} \omega(r)=0$,
which together with monotonicity from Lemma~\ref{lem:flow} implies
\begin{equation}
\begin{aligned}
    \omega(r) \geq 0.
\end{aligned}
\end{equation}
In the two-sided case, the existence of a globally minimal surface means there is some minimal radius $r_0$ where
\begin{equation}
\begin{aligned}
    \omega(r_0) = k r_0^{d-2} + \frac{ r_0^d }{ L^2 },
\end{aligned}
\end{equation}
    which for $k\in\{0, 1\}$ is positive, and so by monotonicity we have $\omega(r)\geq 0$.
\end{proof}

Next, it will be convenient to introduce
\begin{equation}\label{eq:deltapm}
\begin{aligned}
    \delta_{\pm} &\equiv 2\Big|\mathcal{D}_t \phi \mp \frac{ r }{ L }\mathcal{D}_r \phi \Big|^2 
    + 2\left( k - \frac{ \omega(r) }{ r^{d-2}
        }\right)|\mathcal{D}_r \phi|^{2}+ 2 V(\phi, \phi^{\dag})
        \\ 
        & \qquad +\frac{ d }{ d-1 }\mathcal{K}(r)^2 + \frac{ 1 }{ B(r) }F_{tr}^2 + \frac{ 1 }{ 2
        }F^{ij}F_{ij},
\end{aligned}
\end{equation}
so that the constraint \eqref{eq:constraintEqsE} can be written
\begin{equation}
    \begin{aligned}\label{eq:omegadelta}
    (d-1)\frac{ \omega' }{ r^{d-1} } = \pm \frac{ 2r }{ L }J_{r} + \delta_{\pm}(r). \\
\end{aligned}
\end{equation}
In terms of $\delta_{\pm}$ we can then state the following lemma:
\begin{lem}\label{lem:masterlem}
    Consider the solution \eqref{eq:scalarSols} with spherical symmetry, $d\geq 3$ 
    and $V(\phi, \phi^{\dag})\geq 0$, so that the WCC holds.
    Assume that 
    \begin{equation}
        1 - \frac{ \omega(r) }{ r^{d-2} } 
    \end{equation}
    vanishes somewhere. Then there exists a radius $\hat{r}$ such that
    \begin{align}
        1 - \frac{ \omega(\hat{r}) }{ \hat{r}^{d-2} } &=0,  \label{eq:lemeq1}\\
        2\mathcal{E}(\hat{r})+\frac{ d }{ d-1 }\mathcal{K}(\hat{r})^2 &\leq \frac{ (d-1)(d-2) }{ \hat{r}^2 },  \label{eq:lemeq2}\\
        \delta_{\pm}(r) &\geq 0  \qquad \forall r\geq \hat{r}, \label{eq:lemeq3}
    \end{align}
    and 
    \begin{equation}\label{eq:Jintlem}
    \begin{aligned}
        \pm \int_{\hat{r}}^{\infty}\dd r  r^{d}J_{r}(r) = \frac{ (d-1)L }{ 2 }\left[\omega(\infty)-\hat{r}^{d-2} - \frac{ 1
        }{ d-1 }\int_{\hat{r}}^{\infty} \dd r r^{d-1}\delta_{\pm}(r)\right].
    \end{aligned}
    \end{equation}
\end{lem}
\begin{proof}
    Denote for convenience $P(r) =1 - \frac{ \omega(r) }{ r^{d-2} }$.
Since $\omega(r)$ converges to a finite positive number at $r=\infty$ and $d>2$, $P$ is everywhere positive above some large radius.
Since $P$ is negative somewhere, we have by continuity that there must be a last zero of $P$ that is
    approached from negative $P$, so that $P'>0$ there. Let us denote this zero by $r=\hat{r}$.

    Writing \eqref{eq:constraintEqsE} in terms of $P$, we find
\begin{equation}
    \begin{aligned}\label{eq:Prime}
        P'(r) = \frac{ (d-2) }{ r }(1-P) - \frac{ r }{ d-1
    }\left[2 \mathcal{E} + \frac{ d }{ d-1 }\mathcal{K}(r)^2 \right],
\end{aligned}
\end{equation}
which through $P'(\hat{r})\geq0, P(\hat{r})=0$ implies
\begin{equation}
\begin{aligned}
    \frac{ d-2 }{ \hat{r} } - \frac{ \hat{r} }{ d-1 }\left[2\mathcal{E}(\hat{r}) + \frac{ d }{ d-1 }\mathcal{K}(\hat{r})^2
    \right] \geq 0,
\end{aligned}
\end{equation}
giving \eqref{eq:lemeq2}.

Next, note that a single coordinate system is valid for all $r\geq \hat{r}$ since $\frac{ 1 }{ B(r)
} = r^2 + P$ must be strictly positive there.  Then, integrating \eqref{eq:omegadelta} from
$r=\hat{r}$ up to $r=\infty$ and remembering that $\omega(\hat{r}) = \hat{r}^{d-2}$, we find
\begin{equation}
\begin{aligned}
   (d-1)\left[\omega(\infty)-\hat{r}^{d-2}\right] = \int_{\hat{r}}^{\infty} \dd r r^{d-1} \left[ \pm
   \frac{ 2r }{ L }J_{r}+\delta_{\pm}(r)\right],
\end{aligned}
\end{equation}
yielding \eqref{eq:Jintlem}.

Finally, note that the only potentially negative term in $\delta_{\pm}$ is $2P|\mathcal{D}_r \phi|^2$. However, since
$P$ is positive above $\hat{r}$, so is $\delta_{\pm}$, giving \eqref{eq:lemeq3}.
\end{proof}

\subsubsection*{Proving a Lloyd's bound in Einstein-Maxwell-Scalar theory}
We now finally prove our bound on complexity growth (Theorem \ref{thm:Lloyd}):
\begin{proof}
    Assume first that $P=1-\frac{ \omega(r) }{ r^{d-2} }$ is not everywhere positive. The two-sided case will always
be this category, since at a minimal surface, which must always be present on a complete slice, we
have $P=-r^2$. Then by Lemma~\ref{lem:masterlem} there exists an $r=\hat{r} > 0$ such that 
\begin{equation}\label{eq:masterineq}
\begin{aligned}
    \pm \int_{\hat{r}}^{\infty}\dd r  r^{d}J_{r}(r) \leq \frac{ (d-1)L }{ 2 }[\omega(\infty)-\hat{r}^{d-2}].
\end{aligned}
\end{equation}
    Adding $\pm \mathcal{K}(\hat{r})\hat{r}^{d}$ to both sides of the inequality, we see from the solution of the
    constraint \eqref{eq:scalarSols} that we get
    \begin{equation}
    \begin{aligned}
        \pm \lim_{r \rightarrow \infty}r^{d}\mathcal{K}(r) \leq \frac{ (d-1)L }{ 2 }\left[\omega(\infty) - \hat{r}^{d-2} \pm
        \frac{ 2 }{ (d-1)L }\mathcal{K}(\hat{r})\hat{r}^{d}\right] \\
    \end{aligned}
    \end{equation}
    Let us now without loss of generality assume that $\mathcal{K}(\hat{r})\geq 0$ -- the proof of the opposite sign is entirely
    analogous. Taking the lower sign inequality, multiplying by $-\frac{\Omega_{+1}}{G_NL(d-1)}$, neglecting
    the two
    positive terms proportional to $\mathcal{K}$ and $\hat{r}^{d-2}$, and using \eqref{eq:cdotk} and \eqref{eq:omegaM}, we
    get
    \begin{equation}
    \begin{aligned}
        \dCV \geq -\frac{ 8\pi M }{ d-1 }.
    \end{aligned}
    \end{equation}

    Next, \eqref{eq:lemeq2} together with the positivity of $\mathcal{E}$ from the WEC gives 
        $\mathcal{K}(\hat{r})^2 \leq \frac{ (d-1)^2(d-2) }{ d \hat{r}^2 }$. Using this, the
    upper sign inequality reads
    \begin{equation}
        \begin{aligned}\label{eq:plusineq}
        \lim_{r \rightarrow \infty}r^{d}\mathcal{K}(r) 
        &\leq \frac{ (d-1)L }{ 2 }\left[\omega(\infty)  +\hat{r}^{d-2}\left(\frac{ 2 }{ L } \sqrt{\frac{ d-2 }{ d
        }}\hat{r} - 1 \right)\right]. \\
    \end{aligned}
    \end{equation}
    By monotonicity of $\omega$ we know that 
    \begin{equation}
    \begin{aligned}
        \hat{r}^{d-2} = \omega(\hat{r}) \leq \omega(\infty).
    \end{aligned}
    \end{equation}
    Thus, if
    \begin{equation}
        \begin{aligned}\label{eq:omegainfbound}
        \frac{ 2 }{ L }\sqrt{\frac{ d-2 }{ d }}\omega(\infty)^{\frac{ 1 }{ d-2 }}-1 \leq 0,
    \end{aligned}
    \end{equation}
    then the second bracket of \eqref{eq:plusineq} is negative, and so we find $\dCV \leq \frac{ 8\pi M }{ d-1 }$ after
    multiplying by $\Omega_{+1}/G_N L(d-1)$. In terms of a mass, the bound \eqref{eq:omegainfbound} reads
    \begin{equation}
    \begin{aligned}
        M \leq \frac{ (d-1)\vol[S^{d-1}]L^{d-2} }{ 16 \pi G_N }\left[\frac{ d }{
            4(d-2) }\right]^{\frac{ d-2 }{ 2 }} \equiv  \hat{M},
    \end{aligned}
    \end{equation}
    where we use the notation $\Omega_{+1}=\vol[S^{d-1}]$.
    If $M$ does not satisfy this bound, we instead neglect the $-\hat{r}^{d-2}$ term in \eqref{eq:plusineq} and use
    $\hat{r}^{d-1} \leq \omega(\infty)^{\frac{ d-1 }{ d-2 }}$.
    Multiplying by the usual factor we get
    \begin{equation}
    \begin{aligned}
        \dCV \leq \frac{ 8\pi M }{ d-1 }\left[1 + \left(\frac{ M }{ \hat{M} }\right)^{\frac{
        1}{ d-2 }}\right].
    \end{aligned}
    \end{equation}
    Since we proved $\dCV \geq -\frac{ 8\pi M }{ d-1 }$ for any mass, the above bound also holds for the absolute value, and so our bound is proven in the case where $P$ vanishes
    somewhere.

    Assume now $P\geq 0$ everywhere, which is only possible in the one-sided case. Then we can cover the whole of
    $\Sigma$ with one coordinate system, and \eqref{eq:masterineq} holds with $\hat{r}=0$, which
    immediately gives $|\dCV| \leq \frac{ 8\pi }{ d-1 }M$ after multiplying an
    overall factor.

    Finally, we note that the proof for real scalars is entirely analogous, and the only change with 
    multiple fields is that $\delta_{\pm}$ contains a linear sum over the various fields.
    Lemma~\ref{lem:masterlem} remains
    true in this case also. Thus the above proof applies equally well to any number of gauge fields
    and scalars. 
\end{proof}

    \subsection{A simple formula for $\dCV$ for matter of compact support}\label{sec:compact}
\begin{thm}
    Consider an AAdS spacetime with spherical or planar symmetry, and let $r$ be the area radius. Let $\Sigma$ be a maximal volume slice and assume that the matter has support only for $r\leq \rho$ on $\Sigma$. Let $\sigma_{\rho}$ be the $r=\rho$ surface. 
    Then 
    \begin{equation}
        \begin{aligned}\label{eq:dCV2}
        \dCV^2 = \frac{ \Area[\sigma_{\rho}]}{ 4 G_N }\frac{ 64\pi \rho }{  (d-1) L^2   }(M -
        M_{\Sigma}[\sigma_{\rho}]),
    \end{aligned}
    \end{equation}
    If in addition the spacetime satisfies the WCC, then
    \begin{equation}\label{eq:dCV2ineq}
    \begin{aligned}
        \dCV^2 \leq \frac{ 16\pi \rho \Area[\sigma_{\rho}] }{(d-1) G_N L^2  }M.
    \end{aligned}
    \end{equation}
\end{thm}
\begin{proof}
    Since matter has compact support, we can explicitly solve \eqref{eq:constraintEqsE} and \eqref{eq:constraintEqsJ} outside the support of the
    matter:
    \begin{equation}
    \begin{aligned}
        \mathcal{K}(r) &= \frac{ \rho^{d} }{ r^{d} }\mathcal{K}(\rho) \\
        \omega(r) &= \omega(\rho) + \frac{ 1 }{ (d-1)^2 }\mathcal{K}(\rho)^2\left(\frac{ \rho }{ r
        }\right)^{d}(r^{d}-\rho^{d}),
    \end{aligned}
    \end{equation}
    Using that
    \begin{equation}
    \begin{aligned}
        \dCV = \frac{ \Omega_{k} }{ G_N L(d-1) }\rho^{d}\mathcal{K}(\rho), \\
        M_{\Sigma}(\rho) =  \frac{ (d-1)\Omega_{k} }{ 16\pi G_N }\omega(\rho),
    \end{aligned}
    \end{equation}
    we find
    \begin{equation}
    \begin{aligned}
        M = \frac{ (d-1) G_N L^2 }{ 16\pi \Omega_{k} }\frac{ 1 }{ \rho^{d} } \dot{\mathcal{C}}^{2} + M_{\Sigma}(\rho),
    \end{aligned}
    \end{equation}
    giving \eqref{eq:dCV2}. If the WCC holds, then we know that $M_{\Sigma}(\rho)\geq 0$ from
    Lemma~\ref{lem:minorlem}, giving \eqref{eq:dCV2ineq}.
\end{proof}
These formulas make reference to bulk quantities, and so they are useful only when working on the gravitational side.
Furthermore, the compact support restriction makes the result less relevant when gauge fields are present. But 
the bound can be useful for other kinds of matter, or when compact support is a good
approximation. Note also that we do not need compact support in spacetime -- only on a spatial slice. 

\section{Discussion}\label{sec:discussion}

The volumes of maximal slices are among the most natural diffeomorphism invariant gravitational observables in AAdS
spacetimes; these are sensitive to the black hole interior and more generally constitute a more fine-grained
gravitational observable than e.g. the areas of extremal surfaces. The Complexity=Volume relation stands to shed light
on significant aspects of the holographic correspondence if the details of the proposal can be made precise and the proposal can be rigorously established (see~\cite{IliMez21} for some steps in this direction).

Here we have investigated the consistency of CV from a bulk perspective: if the proposal is a fundamental entry
in the holographic dictionary, it dictates constraints on the behavior of maximal volume slices that should be provable
independently using just geometry, analogous to the geometric proof of strong
subadditivity of holographic entanglement entropy~\cite{HeaTak, Wal12}, entanglement wedge nesting~\cite{Wal12}, causal wedge inclusion~\cite{Wal12}, etc. Under an interpretation of $\mathcal{C}_F$ as complexity with the vacuum as reference state, vacuum-subtracted volumes must be strictly positive in all spacetimes not identical to pure AdS. We established this
result rigorously in broad generality in four bulk dimensions, assuming the weak
energy condition. We have also established a weaker statement in other dimensions, that the
vacuum-subtracted volume is positive in spacetimes with sufficient symmetry or in perturbations of the AdS vacuum \cite{HuJi16} (again
assuming the weak energy condition). The more general statement for arbitrary spacetimes satisfying the weak energy
condition would follow from a modification of a well-known mathematical conjecture
\cite{Schoen,Bra97} from
compact to conformally compact manifolds. 

Until now, broadly applicable results on maximal volume slices in holography have been
sparse in comparison with those on (quantum)
extremal surfaces.\footnote{Although not entirely absent: see~\cite{CouEcc18} for example.} This gap is at least partly a consequence of relatively few available techniques for maximal volumes;
the holographic entanglement entropy proposal benefited from a readily-available arsenal of geometric tools controlling
the behavior of codimension-two surfaces long predating holography~\cite{Pen64,HawPen70,HawEll}. Here we initiated the construction of a similar toolbox for maximal volumes,  adapting results from mathematics~\cite{Cho14, BreCho14, AndMet07, AndEic10} in four dimensions and developing new techniques in general $D$. The utility of our technology is immediate: beyond the positive complexity theorem in four dimensions, the new tools have given a derivation of a version of a Lloyd's bound for spatially symmetric maximal volume slices in a large class of matter, which is thus far the broadest proof of the bound on holographic
complexity growth. We have also shown that wormhole complexity is indeed bounded
from below by the thermofield double (with given energy) in general spacetimes
satisfying the weak energy condition. 

It would be interesting to see if our methods could be
refined to derive stronger bounds on complexity growth when charge is present, given that our proof amounted to writing a formula
$(d-1)|\dot{\mathcal{C}}| = 8\pi M - \Delta$ where there pure gauge field terms always contribute positively to
$\Delta$.\footnote{More precisely, we have $\pm(d-1)\dot{\mathcal{C}}_V = 8\pi M - \Delta_{\pm}$ where $\Delta_{\pm}$ is
proportional to the intergral of the quantity $r^{d-1}\delta_{\pm}$ given by \eqref{eq:deltapm}. } A strictly positive lower
bound on their contribution to $\Delta$ in terms of charge would give a
strengthening of our bound. 
It is also a possibility that our bound on complexity growth is not
    strict, and that the bound is true with $f(M)=1$. Yet another possible avenue for further exploration is in spacetimes of
different asymptotics, such as those of \cite{BalKou08,KacXia08,Tay08,Son08,DonHar12}.

The ubiquity of the weak energy condition as an assumption in our theorems raises a potential question: why must we exchange the null energy condition, typically used to prove consistency in the holographic entropy context, for the more restrictive weak energy condition? As we will show in a companion article, the weak energy condition is in fact necessary, and violations of it such as vacuum decay can indeed result in negative vacuum-subtracted complexity, even in wormhole geometries. A potential conclusion is that the gap between the weak energy condition and the
null energy condition is sourced by classical matter whose fine-grained properties in holography are qualitatively different
from weak energy condition respecting fields in ways that are less obvious in coarser observables such as entropy. In particular, $\mathcal{C}_F$ as currently defined cannot be reinterpreted as the complexity
itself for spacetimes with such matter (consisting of e.g. tachyonic scalars satisfying the
Breitenlohner-Freedman bound.)

The assumption of hyperbolicity in our theorems is a more innocuous one in AdS/CFT, although it is nevertheless possible
to violate it via the inclusion of 
compact dimensions. That is, an
 asymptotically AdS$_{d+1}\times K$
spacetime, $K$ can be picked to be a compact space that spoils the assumption of hyperbolicity even for static slices of
pure AdS (when the compact dimensions are included). In an upcoming paper \cite{EngFol21b} we show that when this happens (which can be the case in many supergravity theories of interest), then the
comparison theorems can be violated even when the weak energy condition holds, resulting in negative ${\cal C}_{F}$. This suggests that the effects of compact
dimensions on the CV proposal are nontrivial and deserve further study. 

Let us now briefly discuss a potential application of our results to the open question of holographic complexity as applied to subregions.

\paragraph{Mixed state holographic complexity:} Given a reduced density matrix $\rho_R$, what is the holographic dual of the least complex purification of $\rho_{R}$  (under the constraint of no unentangled
qubits in the purifier)? This question was initially asked by \cite{AgoHea18}, who dubbed the corresponding quantity the
purification complexity $C_{P}(\rho_{R})$.\footnote{It was proposed by~\cite{Ali15} that the relevant geometric quantity
is the volume of $\Sigma_{E_{R}}$, the maximal volume slice of the entanglement wedge of $R$. However, this proposal
falls short of satisfying the requisite qualitative properties predicted by tensor network models \cite{AgoHea18}.} The
natural bulk dual to this quantity would require a minimax procedure: consider the maximal volume slices of all possible
classical\footnote{It is of course possible that the optimal purification does not have a semiclassical dual. However,
at least in some cases -- e.g. for optimized ($n$-point) correlation measures -- it was established in~\cite{Che19} that
a minimization over semiclassical bulk spacetimes will in fact accomplish a global minimum over a holographic CFT's
Hilbert space.} bulks that complete the entanglement wedge of $R$ into an inextendible spacetime, and then minimize its
volume over the spacetime geometries. The relevant question then appears to be: given $\rho_{R}$, what is the spacetime of least complexity containing the entanglement wedge of $R$?

The volume of the maximal volume slice of
this spacetime is the obvious candidate to $\mathcal{C}_P(\rho_R)$.\footnote{
    This way of computing $\mathcal{C}_P(\rho_R)$ was proposed in
\cite{CacCou18}, although they also minimized the complexity over bulk cutoffs, allowing the bulk cutoff to move deep
into the bulk. Then then solution is that $\Sigma_{E_R}$ itself is the Cauchy slice of
the purified bulk spacetime, with the conformal boundary moved in sufficiently far that the HRT surface itself now corresponds to a
piece of the conformal boundary. This leads to a peculiar situation where the state proposed to
purify $\rho_R$ is a state on a QFT where the UV cutoff drastically differs in different regions
of space. This appears in tension with purity, since the theory on the HRT surface must purify the
state on the real conformal boundary, even as the UV cutoff is taken arbitrarily small and the
number of qubits there grows without bound. Avoiding the minimization over the cutoff altogether appears to be the more natural option.
}
The results from Sec.~\ref{sec:lowerbound} allow for
a concrete computation. Consider the case of a two-sided
spherically symmetric connected spacetime, and take $R$ to be a complete timeslice of one of the two conformal
boundaries. Assuming the weak energy condition, a computation carried out in Appendix~\ref{sec:diamondappendix} proves that the spherically symmetric spacetime of least
$\mathcal{C}_V$ complexity that contains $E_R$ is a one-sided spacetime with $\mathcal{C}_V$ complexity
\begin{equation}\label{eq:CPcomp}
\begin{aligned}
\frac{ V(\Sigma_{E_R}) }{ G_N L } +S_{\mathrm{vN}}(\rho_R) \frac{
            4 r  }{ d L } {}_2 F_{1}\left(\frac{ 1 }{ 2 }, \frac{ d }{ 2 }, \frac{ 2+d }{ 2 }, - \frac{ r^2 }{
            L^2 } \right),
\end{aligned}
\end{equation}
where $r$ is the area radius of the HRT surface.\footnote{We cannot rule out that there is a spacetime that breaks
spherical symmetry with even lower complexity.}  
The above is just the volume obtained after gluing a ball of the hyperbolic plane to the maximal volume slice of the
entanglement wedge, so that the resulting surface is a complete initial data set (with $K=0$). See
Fig.~\ref{fig:purification} for an illustration.
\begin{figure}
\centering
\includegraphics[width=1.0\textwidth]{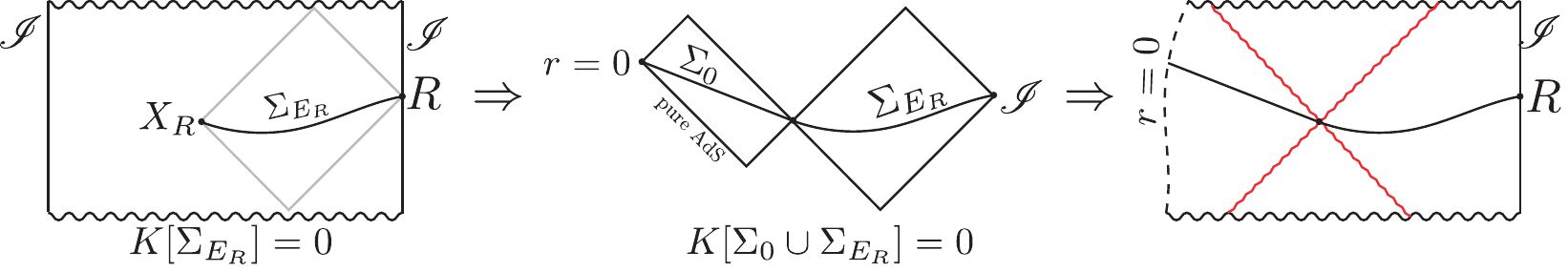}
    \caption{Example of the proceedure used to build the simplest WCC respecting spherically symmetric spacetime  (at time $R$) containing the entanglement wedge $E_R$ of
    $R$. The extremal Cauchy slice $\Sigma_{E_R}$ of an entanglement wedge $E_R$ (left) is glued to a ball
    $\Sigma_0$ of the hyperbolic plane with extrinsic curvature $K_{ab}=0$ (center). Then the maximal Cauchy evolution
    of $\Sigma_{E_R}\cup \Sigma_0$ gives the wanted spacetime (right). The red wiggly lines describe positive energy
    shocks.}
\label{fig:purification}
\end{figure}

Ref. \cite{AgoHea18} argued for several qualitative features of $\mathcal{C}_{P}$ that are satisfied by~\eqref{eq:CPcomp}. First, $\mathcal{C}_P(\rho_R)$ is expected to go as  $c_1 n + |c_2| S(\rho_R)$ for
unknown $c_i$, where $n$ is the number of qubits in the state. Since
the leading UV-divergence of the volume of $\Sigma_{E_R}$ is of the form $\vol[R]/\epsilon^{d-1}$,
which can be thought of as counting the number of lattice sites in the discretized theory, the two terms in
Eq.~\eqref{eq:CPcomp} have this form. Second, $\mathcal{C}_{P}$ was argued to satisfy subadditivity for the thermofield
double state: $2 \mathcal{C}_P(\rho_{\beta}) > \mathcal{C}(\ket{\mathrm{TFD}, \beta})$. Again, this is
realized by our result when applied to one side of a Schwarzschild black hole, 
unlike in the case where the second term of \eqref{eq:CPcomp} is not included.

We can make further progress on this front by making the assumption that the optimal purification is always given by
completing $\Sigma_{E_R}$  with a compact subset $H$ of the hyperbolic plane.\footnote{Of course we should prescribe $K_{ab}$ to have initial data. One can show that
gluing a subset of the hyperbolic plane to an extremal surface respects the WCC and the NEC
if we assign $K_{ab}=0$ to $H$. If the AdS WCC does not hold but there is still a lower bound on
energy density in our theory, we should then choose hyperbolic space with the smallest cosmological constant
consistent with our minimal energy density, so as to minimize its volume.} This assumption is motivated by the
conformally compact generalization of Schoen's conjecture and the following theorem (together with the fact that
hyperbolic space has constant sectional curvature sectional curvature $-\frac{ 1 }{ L^2 }$):
\begin{thm}[\cite{Yau75}]
Let $\Sigma$ be a compact domain in a $d$-dimensional complete simply connected Riemannian manifold with sectional curvature bounded from above by
$-\kappa^2$. Then
\begin{equation}
\begin{aligned}
    \vol[\Sigma] \leq \frac{ \Area[\partial \Sigma] }{ (d-1) |\kappa| }.
\end{aligned}
\end{equation}
\end{thm}
This immediately yields the following lower and upper bounds on purification complexity:
\begin{equation}
\begin{aligned}
    \frac{ \vol[\Sigma_{E_R}] }{ G_N L } \leq \mathcal{C}_P(\rho_{R}) \leq \frac{ \vol[\Sigma_{E_R}] }{ G_N L } + \frac{ 4 }{
    d-1 }S_{\mathrm{vN}}(\rho_{R}).
\end{aligned}
\end{equation}
Note that if the WCC is violated but there is still a lower bound on local energy density, we could get a similar bound,
but with an altered prefactor in the entropy term.

\section*{Acknowledgments}
It is a pleasure to Chris Akers, Simon Brendle, Shira Chapman, Otis Chodosh,
Sebastian Fischetti, Daniel Harlow,
Daniel Roberts, Shreya Vardhan and Ying Zhao for discussions.  This work is supported in part by NSF grant no.
PHY-2011905 and the MIT department of physics. The work of NE was also supported in part by the U.S. Department
of Energy, Office of Science, Office of High Energy Physics of U.S. Department of Energy under grant Contract Number
DE-SC0012567 (High Energy Theory research) and by the U.S. Department of Energy Early Career Award DE-SC0021886. The work of \AA{}F is also
supported in part by an Aker Scholarship. 

\appendix
\section{Appendix}
\toclesslab\subsection{Asymptotic hyperbolicity}{sec:asymhyp}
Let now $(\tilde{\Sigma}, \tilde{h})$ be hyperbolic $3$-space, with metric
\begin{equation}
\begin{aligned}
\dd s^2 = \frac{ 1 }{ 1+r^2 }\dd r^2 + r^2 \dd \Omega^2.
\end{aligned}
\end{equation}
Consider a second Riemannian $3$-manifold $(\Sigma, h)$, potentially with boundary,
where $\Sigma = \tilde{\Sigma} \setminus K$ for a compact (possibly empy) set $K$. Then $(\Sigma, h)$ is in Ref.~\cite{BreCho14} defined to be
asymptotically hyperbolic if
\begin{equation}
\begin{aligned}
        |h-\tilde{h}|_{\tilde{h}} &= \mathcal{O}(r^{-2-\delta}), \quad \delta >0, \\
        \lim_{r \rightarrow \infty} |\tilde{D}(h-\tilde{h})|_{\tilde{h}} &= 0,
\end{aligned}
\end{equation}
where $|\cdot |_{\tilde{h}}$ is the pointwise tensor norm taken with respect to $\tilde{h}$.

Consider now an AAdS$_4$ spacetime $(M, g)$ and pick a boundary conformal frame so that 
\begin{equation}
\begin{aligned}
\dd s^2|_{\partial M} = - \dd t^2 + \dd \Omega^2.
\end{aligned}
\end{equation}
We proceed now to show that an extremal hypersurface $\Sigma$ in an AAdS$_4$ spacetime anchored at $t=\text{const}$ slice
on the boundary is asymptotically hyperbolic. 

First, note that the Fefferman-Graham coordinates of pure AdS$_4$ adapted to the Einstein static universe of unit spatial radius on the boundary reads
\begin{equation}
\begin{aligned}
    \dd s^2 = \frac{ 1 }{ z^2 }\left[\dd z^2 - \frac{ (4+z^2)^2 }{ 16 }\dd t^2 + \frac{ (4-z^2)^2 }{ 16 }\dd \Omega^2
    \right].
\end{aligned}
\end{equation}
Going to a general AAdS$_4$ spacetime with the falloffs we are considering, we have
\begin{equation}
\begin{aligned}
    \dd s^2 = \frac{ 1 }{ z^2 }\left[\dd z^2 - \left(1 + \frac{ z^2 }{ 2 } \right)\dd t^2 + \left(1 - \frac{ z^2 }{ 2 } \right)\dd \Omega^2
    + z^3 \mathcal{T}_{ij}(x) \dd x^{i}\dd x^{j} + \mathcal{O}(z^4)\right],
\end{aligned}
\end{equation}
where $i, j$ are boundary coordinate indices running over $t, \theta, \varphi$.

Consider now a spatial hypersurface $\Sigma$ in $(M, g)$ that is extremal, and with intrinsic coordinates $x^{\alpha}=(z, \theta,
\varphi)$ and embedding coordinates $X^{\mu}=(z, t(z, \theta, \varphi), \theta, \varphi)$. Generically, once the boundary
anchoring time is fixed, there is only one slice that will be
smooth and complete in the bulk, and so the integration constant which determines how $\Sigma$ leaves the boundary is fixed
once the boundary anchoring location of $\Sigma$ is fixed. However, from a near boundary analysis it is impossible to know
this integration constant, so we will have to allow it to be general.

Let us take the ansatz for the expansion of $t(z, \theta, \varphi)$ near the boundary to be given by
\begin{equation}
\begin{aligned}
    t(z, \theta, \varphi) = \sum_{n=0}^{\infty} t_n(\theta, \varphi) z^{n}.
\end{aligned}
\end{equation}
Since we consider a slice of constant $t$ on the boundary, we have that $t_0$ is a constant. 
With this expansion we can compute the mean curvature $K$ in a small$-z$ expansion and demand that it vanishes order by
order in $z$. The exact form of the equations are ugly and not needed. We only need the basic structure. We find that
\begin{equation}
\begin{aligned}
    t_1(\theta, \varphi) = t_2(\theta,\varphi) = t_3(\theta, \varphi) = 0.
\end{aligned}
\end{equation}
If the anchoring time $t_0$ was not constant, then $t_2$ would become nonzero and given by an algebraic expression of the derivatives of $t_0$. 
The function $t_4$ is the integration constant referred to above. It encodes deep bulk information
and is fixed by requiring that $\Sigma$ is a smooth slice. Higher $t_n$ are
fixed by $t_0$, $t_4$, $\mathcal{T}_{ij}$ and higher order terms in the metric (which in turn depends only on the boundary
conformal structure and $\mathcal{T}_{ij}$). 

The induced metric on $\Sigma$ now reads
\begin{equation}
\begin{aligned}
    h_{\alpha\beta} &= g_{\alpha\beta} +  g_{tt} \partial_{\alpha}t \partial_{\beta}t \\
    &= \frac{ 1 }{ z^2 } \begin{pmatrix}
    1 & 0 & 0 \\
    0 & (1 - \frac{ z^2 }{ 2 }) + z^3 \mathcal{T}_{\theta\theta} & \mathcal{T}_{\theta\phi}z^3 \\
    0 & \mathcal{T}_{\phi\theta}z^3 & (1-\frac{ 1 }{ 2 }z^2) \sin\theta^2 + \mathcal{T}_{\phi\phi}z^3 \\
\end{pmatrix} 
    + \mathcal{O}(z^2).
\end{aligned}
\end{equation}
In the above, higher order terms depend both on the spacetime geometry and $t_4$. Since the $t_4$
dependence is in the higher order term, setting $\mathcal{T}_{ij}=0$ gives the metric of
the hyperbolic plane, $\tilde{h}_{\alpha\beta}$, up to
$\mathcal{O}(z^2)$ corrections. Defining for notational convenience
$\mathcal{T}_{\alpha\beta}=\mathcal{T}_{ij}$ when $\alpha, \beta, i, j \in \{\theta,
\varphi \}$ and $\mathcal{T}_{z\alpha}=0$, we have
\begin{equation}
\begin{aligned}
    |h-\tilde{h}|^2_{\tilde{h}} &= \tilde{h}^{\alpha\gamma}\tilde{h}^{\beta\delta}(h-\tilde{h})_{\alpha\beta}(h-\tilde{h})_{\gamma\delta} \\
                            &= \frac{ 1 }{ z^4 }\tilde{h}^{\alpha\gamma}\tilde{h}^{\beta\delta} \left(z^3
                            \mathcal{T}_{\alpha\beta}+\mathcal{O}(z^4)\right) \left( z^3 \mathcal{T}_{\gamma\delta} + \mathcal{O}(z^4)
                            \right) \\
                            &=\mathcal{O}(z^{6}).\label{eq:hdiff}
\end{aligned}
\end{equation}
But to leading order, we have $z=\frac{ 1 }{ r } + \mathcal{O}(r^{-2})$, where $r$ is the coordinate used to define
asymptotic hyperbolicity. Thus we find
\begin{equation}
\begin{aligned}
    |h-\tilde{h}|_{\tilde{h}} = \mathcal{O}(r^{-3}),
\end{aligned}
\end{equation}
meaning that $\Sigma$ satisfies the first condition for asymptotic hyperbolicity. Next, note that
\begin{equation}
\begin{aligned}
    \tilde{D}_{\gamma}(h-\tilde{h})_{\alpha\beta} = \tilde{D}_{\gamma}\left[ z \mathcal{T}_{\alpha\beta}(x) + \mathcal{O}(z^2) \right] =
    \mathcal{O}(1).
\end{aligned}
\end{equation}
Since the three inverse metrics involved in calculating $|\tilde{D}(h-\tilde{h})|^2$ brings a total power of $z^6$, we find that 
\begin{equation}
\begin{aligned}
    |\tilde{D}(h-\tilde{h})|  = \mathcal{O}(z^3) = \mathcal{O}(r^{-3})
\end{aligned}
\end{equation}
and so the second condition,
\begin{equation}
\begin{aligned}
    \lim_{r \rightarrow \infty} |\tilde{D}(h-\tilde{h})| =0, 
\end{aligned}
\end{equation}
holds. Hence $\Sigma$ is asymptotically hyperbolic.

Finally we note that if $t_0$ was not constant, then $t_2$ would not vanish, and we would have $\mathcal{O}(1)$ corrections in $h_{zz}$ depending on
$t_2$. This factor would not be present for the metric of hyperbolic space, and so $(h-\tilde{h})_{\alpha\beta}$ would now
be $\mathcal{O}(z^2)$ rather than $\mathcal{O}(z^3)$, and so the falloff in Eq.~\eqref{eq:hdiff} would end up being $\mathcal{O}(z^4)$ instead, which would mean $\Sigma$ was not asymptotically
hyperbolic in general. 

\toclesslab\subsection{$d$-dimensional Geroch-Hawking-mass with a cosmological constant}{sec:GerHawMass}
Consider a maximal volume slice $\Sigma$ of a spacetime with maximal spatial symmetry.
The mean curvature of a constant$-r$ surface $\sigma$ in $\Sigma$ in the coordinate system~\eqref{eq:cancoords} reads
\begin{equation}
\begin{aligned}
    H = D_{\alpha}r^{\alpha} = \frac{ d-1 }{ r\sqrt{B(r)} },
\end{aligned}
\end{equation}
where $r^{\alpha}= \frac{ 1 }{ \sqrt{A} }(\partial_r)^{\alpha}$ is the unit normal pointing to increasing $r$. 
We find
\begin{equation}
\begin{aligned}
    H^2 &= \frac{ (d-1)^2 }{ r^2 }\left(k + \frac{ r^2 }{ L^2 } - \frac{ \omega(r) }{ r^{d-2}
    }\right).
\end{aligned}
\end{equation}
Noting that a constant$-r$ surface has intrinsic Ricci scalar 
\begin{equation*}
\begin{aligned}
    \mathcal{R} = k \frac{ (d-1)(d-2) }{ r^2 },
\end{aligned}
\end{equation*}
we can rewrite $\omega$ as follows:
\begin{equation*}
\begin{aligned}
    \omega &= r^{d}\left[\frac{ \mathcal{R} }{ (d-1)(d-2) } - \frac{ H^2 }{ (d-1)^2 } + \frac{ 1
    }{L^2 }\right] \\
     &= \left(\frac{ A[\sigma] }{ \Omega_{k} }\right)^{\frac{ d }{ d-1 }}\frac{ 1 }{ A[\sigma]
         }\int_{\sigma}\bm{\epsilon}\left[\frac{ \mathcal{R} }{ (d-1)(d-2) } - \frac{ H^2 }{ (d-1)^2 } + \frac{ 1
    }{L^2 }\right] \\
     &= \frac{ 1 }{ \Omega_{k} }\left(\frac{ A[\sigma] }{ \Omega_{k} }\right)^{\frac{ 1 }{ d-1
             }}\int_{\sigma}\bm{\epsilon}\left[\frac{ \mathcal{R} }{ (d-1)(d-2) } - \frac{ H^2 }{ (d-1)^2 } + \frac{ 1
    }{L^2 }\right]
\end{aligned}
\end{equation*}
where $\bm{\epsilon}$ is the volume form on the constant$-r$ surface $\sigma$.
This is a covariant functional $\omega[\sigma, \Sigma]$. In fact, if multiply by the overall
constant $\Omega_k^{\frac{ d }{ d-1 }}$, there is no reference to which symmetry we have. 

\toclesslab\subsection{Monotonicity of the volume of Stat$_k$ with respect to mass}{app:Ik}
Assume $d>2$ and consider the regularized volume of (half a) totally geodesic slice of Stat$_k$:
\begin{equation}
\begin{aligned}
    V(r_h) = \int_{r_h}^{r_c}\dd r r^{d-1}\left[ k + r^2 - \left(\frac{ r_h }{ r }\right)^{d-2} (k +
        r_h^2) \right]^{-1/2},
\end{aligned}
\end{equation}
where we pick units where $L=1$ and divide out the overall factor of $\Omega_k$. Changing now variables
\begin{equation}
\begin{aligned}
    r = e^{t}r_h \qquad \dd r = r \dd t, \quad \tau \equiv t(r_c) = \log\left(\frac{ r_c }{ r_h }
    \right),
\end{aligned}
\end{equation}
we get
\begin{equation}
\begin{aligned}
    V(r_h) &= r_h^{d}\int_{0}^{\tau}\dd t e^{d t}\left[k + e^{2t}r_h^2 - e^{-(d-2)t}(k + r_h^2)
    \right]^{-1/2} \\
           &= r_h^{d}\int_{0}^{\tau}\dd t e^{(d-1)t}\left[k\left(e^{-2t} - e^{-dt}\right) + r_h^2 \left(
                   1 - e^{-dt } \right) \right]^{-1/2}
\end{aligned}
\end{equation}
and
\begin{equation}
\begin{aligned}
    V(e^{\alpha}r_h) &= r_h^{d}\int_{0}^{\tau-\alpha}\dd t e^{d(t + \alpha)}\left[k\left(1 -
            e^{-(d-2)t}\right) + r_h^2 e^{2\alpha} \left(
                   e^{2t} - e^{-(d-2)t } \right) \right]^{-1/2} \\
            &= r_h^{d}\int_{\alpha}^{\tau}\dd t e^{d t}\left[k\left(1 -
                    e^{-(d-2)(t-\alpha)}\right) + r_h^2 e^{2\alpha} \left(
            e^{2(t-\alpha)} - e^{-(d-2)(t-\alpha) } \right) \right]^{-1/2} \\
            &= r_h^{d}\int_{\alpha}^{\tau}\dd t e^{(d-1) t}\left[k\left(e^{-2t} -
                    e^{(d-2)\alpha }e^{-dt}\right) + r_h^2 \left(
            1 - e^{d\alpha}e^{- dt} \right) \right]^{-1/2}. 
\end{aligned}
\end{equation}
Consider now the difference in volume of a slice with horizon $e^{\alpha}r_h$ and one with $r_h$. 
Furthermore, introduce a new regulator $\epsilon \ll 1$ by modifying the denominators above with the
replacement $[\ldots] \rightarrow [\epsilon +
\ldots]$. The (rescaled and $\epsilon$-regulated) volume difference reads
\begin{equation}
\begin{aligned}
    \Delta_{\epsilon}(\alpha, r_h) \equiv \lim_{\tau \rightarrow \infty }\frac{ V_{\epsilon}(e^{
    \alpha}r_h) - V_{\epsilon}(r_h) }{ r_h^{d} }.
\end{aligned}
\end{equation}
We have that $\partial_{r_h}V(r_h)$,  $\partial_m I_{k}(m)$ and 
$\partial_{\alpha}\Delta_{\epsilon}(\alpha, r_h)|_{\alpha=0, \epsilon=0}$ all have the same sign, so
let us focus on the latter.

We have
\begin{equation}
\begin{aligned}
    \Delta_{\epsilon}(\alpha, r_h) &= \int_{\alpha}^{\infty}\dd t e^{(d-1)t}\Big\{\left[\epsilon+ k\left(e^{-2t} -
                    e^{(d-2)\alpha }e^{-d t}\right) + r_h^2 \left(
            1 - e^{d\alpha}e^{- dt} \right) \right]^{-1/2}  \\
        & \qquad  -\left[\epsilon + k\left(e^{-2t} - e^{-d t}\right) + r_h^2 \left(
                   1 - e^{-d t } \right) \right]^{-1/2}
        \Big\} \\
        & - \int_{0}^{\alpha}\dd t e^{(d-1) t}\left[\epsilon + k\left(e^{-2t} - e^{-dt}\right) + r_h^2 \left(
                   1 - e^{-dt } \right) \right]^{-1/2},
\end{aligned}
\end{equation}
and so taking the $\alpha$-derivative we obtain
\begin{equation}
\begin{aligned}
    \partial_{\alpha}\Delta_{\epsilon}(\alpha) = \int_{\alpha}^{\infty}\dd t e^{(d-1)t} \partial_{\alpha}\{\ldots
    \} - e^{(d-1)\alpha} \epsilon^{-1/2}.
\end{aligned}
\end{equation}
Next, note that
\begin{equation}
\begin{aligned}
    \partial_{\alpha}[\ldots]^{-1/2} = \frac12 [\ldots]^{-3/2} d \left(k\frac{d-2}{d} e^{-2\alpha} +
    r_h^2 \right) e^{\alpha d} e^{-d t}.
\end{aligned}
\end{equation}
Defining now $\mu=k\frac{ d-2 }{d} + r_h^2$,  we find that
\begin{equation}\label{eq:deltaint}
\begin{aligned}
    \partial_{\alpha}\Delta_{\epsilon}|_{\alpha=0} & =-\epsilon^{-1/2} 
    + \frac{ d\mu}{ 2  } \int_{0}^{\infty} \dd t e^{-t} \left[\epsilon + \mu \left(
            1 - e^{- dt} \right) + k\left(e^{-2t} -
    e^{-d t}\right) -k\frac{ d-2 }{ d }(1-e^{-dt}) \right]^{-3/2} \\
        & =-\epsilon^{-1/2} 
    + \frac{ d\mu}{ 2  } \int_{0}^{\infty} \dd t e^{-t} \left[\epsilon + \mu \left(
    1 - e^{- dt} \right) + k f(t) \right]^{-3/2} \\
\end{aligned}
\end{equation}
where 
\begin{equation}
\begin{aligned}
    f(t) \equiv e^{-2t}-e^{-dt} - \frac{ d-2 }{ d }\left( 1 -e^{-dt}\right).
\end{aligned}
\end{equation}
Now we must consider various cases.
\subsubsection*{The case of $k\geq 0$:}
For all $t\geq0$ we have that
\begin{equation}\label{eq:band}
\begin{aligned}
    f(t) \leq 0.
\end{aligned}
\end{equation}
Furthermore, for $k\geq 0$ we have $\mu \geq 0$ for all $r_h>0$, and so replacing $f(t)\rightarrow 0$
we get the inequality
\begin{equation}
\begin{aligned}
    \partial_{\alpha}\Delta_{\epsilon}|_{\alpha=0} & \geq  -\epsilon^{-1/2} 
    + \frac{ d\mu}{ 2  } \int_{0}^{\infty} \dd t e^{-t} \left[\epsilon + \mu \left(
            1 - e^{- dt} \right) \right]^{-3/2}.
\end{aligned}
\end{equation}
This intergral is a hypergeometric function whose expansion about $\epsilon=0$ reads
\begin{equation}
\begin{aligned}
    \partial_{\alpha}\Delta_{\epsilon}|_{\alpha=0} & \geq -\epsilon^{-1/2} 
    + \frac{ d\mu}{ 2  }\left[ \frac{ 2 }{ d \mu \sqrt{\epsilon} }- \frac{ 2\sqrt{\pi}\Gamma\left(1+\frac{ 1 }{ d } \right) }{ 
        \mu^{3/2} \Gamma\left(-\frac{ 1 }{ 2 }+\frac{ 1 }{ d } \right) } +
        \mathcal{O}(\sqrt{\epsilon})\right],
\end{aligned}
\end{equation}
giving finally that
\begin{equation}
\begin{aligned}
\partial_{\alpha}\Delta_{\epsilon}|_{\alpha=0, \epsilon=0} \geq - \frac{ d \sqrt{\pi}\Gamma\left(1+\frac{ 1 }{ d } \right) }{ 
        \mu^{1/2} \Gamma\left(-\frac{ 1 }{ 2 }+\frac{ 1 }{ d } \right) } > 0,
\end{aligned}
\end{equation}
where we use that $d>2$ so that the Gamma-function in the denominator is negative. Thus, for
a spherical or planar static black hole we have that the complexity of formation is positive and monotonically
increasing with horizon radius, and thus also mass.

\subsubsection*{The case of $k=-1$:}
Below we only show monotonicity for sufficiently large masses.

Consider $\mu>0$. Then the integral in \eqref{eq:deltaint} can be expanded
\begin{equation*}
\begin{aligned}
    J  &\equiv  \frac{ d }{ 2\sqrt{\mu}}\int_{0}^{\infty}\dd t e^{-t} \left[\epsilon \mu^{-1} + \left(
    1 - e^{- dt} \right) - \mu^{-1} f(t) \right]^{-3/2} \\
       &=  \frac{ d }{ 2\sqrt{\mu}}\sum_{n=0}^{\infty}
       c_n\mu^{-n} \int_{0}^{\infty}\dd t e^{-t} \frac{f(t)^{n}}{\left[\epsilon \mu^{-1} + \left(
       1 - e^{- dt} \right)\right]^{3/2+n}}, 
\end{aligned}
\end{equation*}
where $c_n$ are the positive coefficients appearing in
\begin{equation*}
\begin{aligned}
    \frac{ 1 }{ (1-x)^{3/2} }=\sum_{n=0}^{\infty}c_n x^n.
\end{aligned}
\end{equation*}
Note that $f(t)\sim \mathcal{O}(t^2)$ at small $t$,
so in this limit the numerator behaves as $\sim t^{2n}$ while the denominator behaves as $\sim (\epsilon +
t)^{3/2 + n}$. After carrying out the integral and sending $\epsilon\rightarrow 0$ only the $n=0$ term diverges, so after carrying out the
integral each $n\geq 1$ term
scales as $\mathcal{O}(\epsilon^0)$. Thus up to $\mathcal{O}(\epsilon)$
corrections we have
\begin{equation*}
\begin{aligned}
    J&=  \frac{ d }{ 2\sqrt{\mu}}\int_{0}^{\infty}\dd t e^{-t} \frac{1}{\left[\epsilon \mu^{-1}
       + 1 - e^{- dt}\right]^{3/2}}+
    \sum_{n=1}^{\infty}
       \frac{ c_n d \mu^{-n-\frac{ 1 }{ 2 }}}{ 2}\int_{0}^{\infty}\dd t \frac{e^{-t} f(t)^{n}}{\left[
       1 - e^{- dt}\right]^{3/2+n}}.
\end{aligned}
\end{equation*}
The first term is exactly the integral computed above, so we find
\begin{equation*}
\begin{aligned}
\partial_{\alpha}\Delta_{\epsilon}|_{\alpha=0, \epsilon=0} =  - \frac{ d \sqrt{\pi}\Gamma\left(1+\frac{ 1 }{ d } \right) }{ 
        \mu^{1/2} \Gamma\left(-\frac{ 1 }{ 2 }+\frac{ 1 }{ d } \right) } + \sum_{n=1}^{\infty}
       \frac{ c_n d \mu^{-n-\frac{ 1 }{ 2 }}}{ 2}\int_{0}^{\infty}\dd t \frac{e^{-t} f(t)^{n}}{\left[
       1 - e^{- dt}\right]^{3/2+n}}.
\end{aligned}
\end{equation*}
We can readily check that, for $t\geq 0$, we have the lower bound
\begin{equation*}
\begin{aligned}
    f(t) \geq -\frac{ d-2 }{ d }(1-e^{-d t})^2.
\end{aligned}
\end{equation*}
giving
\begin{equation*}
\begin{aligned}
    f(t)^{n} \geq \begin{cases}
        0 & n \text{ even} \\
        -(\frac{ d-2 }{ d })^{n}(1-e^{-d t})^{2n} & n \text{ odd}.
    \end{cases}
\end{aligned}
\end{equation*}
Thus we find
\begin{equation*}
\begin{aligned}
    \partial_{\alpha}\Delta_{\epsilon}|_{\alpha=0, \epsilon=0} &\geq  - \frac{ d \sqrt{\pi}\Gamma\left(1+\frac{ 1 }{ d } \right) }{ 
    \mu^{1/2} \Gamma\left(-\frac{ 1 }{ 2 }+\frac{ 1 }{ d } \right) } - \sum_{n=1, \mathrm{odd}}^{\infty}
        \frac{ c_n d \mu^{-n-\frac{ 1 }{ 2 }}}{ 2}\left(\frac{ d-2 }{ d }\right)^{n } \int_{0}^{\infty}\dd t e^{-t}\left[
       1 - e^{- dt}\right]^{n-\frac{ 3 }{ 2 }} \\
                                                               &=  - \frac{ d \sqrt{\pi}\Gamma\left(1+\frac{ 1 }{ d } \right) }{ 
        \mu^{1/2} \Gamma\left(-\frac{ 1 }{ 2 }+\frac{ 1 }{ d } \right) } - \sum_{n=1, \mathrm{odd}}^{\infty}
        \frac{ c_n d \mu^{-n-\frac{ 1 }{ 2 } }}{ 2}\left(\frac{ d-2 }{ d }\right)^{n }
        \frac{ \Gamma\left(1+\frac{ 1 }{ d }\right)\Gamma\left(-\frac{ 1 }{ 2 }+n\right) }{
        \Gamma\left(-\frac{ 1 }{ 2 }+\frac{ 1 }{ d } + n \right)}.
    \end{aligned}
\end{equation*}
This last sum can be carried out exactly and evalutes to a hypergeometric function, but the exact expression is not particularly useful.
However, it does scale like $\mathcal{O}(\mu^{-3/2})$ at large $\mu$, so since the first term which
goes as $\mathcal{O}(\mu^{-1/2})$ has a positive coefficient, this means that there
exists a $\hat{\mu}>0$ such that $\partial_{\alpha}\Delta_{\epsilon}|_{\alpha=0, \epsilon=0}\geq 0$ for all
$\mu\geq \hat{\mu}$.

\toclesslab\subsection{Computing $\dot{\mathcal{C}}$}{app:Cdot}
We study only spacetimes that have finite gravitational mass, so that the
mass \eqref{eq:omegaM} coincides the CFT energy up to the Casimir energy. In the coordinates \eqref{eq:cancoords}
this means that the following integrals must be assumed finite: 
\begin{equation}
\begin{aligned}
    \int_{}^{\infty} \dd r r^{d-1} \mathcal{E}, \qquad \int_{}^{\infty} \dd r r^{d} J_r(\rho),
\end{aligned}
\end{equation}
where $J_r$ is given by
\begin{equation}
\begin{aligned}
    J_r = 8\pi G_N n^a (\partial_r)^b T_{ab},
\end{aligned}
\end{equation}
and where $8 \pi G_N T_{ab} = R_{ab}-\frac{ 1 }{ 2 }g_{ab}R - \frac{ d(d-1) }{ 2L^2 }g_{ab}$.

We now want to express $\eta^a N_a$, as defined in Sec.~\ref{sec:easycdot}, in terms of quantities over which we have control. To do this we temporarily introduce an ADM coordinate system $x^{\mu}$ on spacetime with vanishing shift, and lapse equal to $r$:
\begin{equation}
\begin{aligned}
    \dd s^2|_{M} = g_{\mu\nu}\dd x^{\mu}\dd x^{\nu} = -r^2 \dd t^2 + h_{\alpha\beta}(t, x^{\alpha}) \dd x^{\alpha}\dd x^{\beta},
\end{aligned}
\end{equation}
where
\begin{equation}
\begin{aligned}
    \dd s^2|_{\Sigma} &=  h_{\alpha\beta}(t=0, x^{\alpha}) \dd x^{\alpha}\dd x^{\beta} = B(r) \dd r^2 + r^2 \dd
    \Omega_{k}^2, \\
    \partial_{t}h_{\alpha\beta}(x)|_{t=0} &= 2 r K_{\alpha\beta}(x).
\end{aligned}
\end{equation}
We proceed to locate the conformal boundary in a small neighborhood around $\Sigma$. This will let us determine
$\eta^a$, which is tangent to $\partial M$.  

Let us for notational convenience now take the spherically symmetric case, and let us install coordinates 
$y^{i} = (\tau, \Omega)$ on $\partial M$ and take $\partial M \cap \Sigma = \partial \Sigma$ to be located
$(r=r_c, t=\tau =0)$, where we temporarily work with a finite cutoff at $r=r_c$. We want to find embedding coordinates $(r(\tau),
t(\tau))$ for $\partial M$ so that the induced metric reads
\begin{equation}
\begin{aligned}
    \dd s^2|_{\partial M} = \gamma_{ij}\dd y^{i}\dd y^{j} = \frac{r_c^2}{L^2} (-\dd \tau^2 + L^2 \dd \Omega^2).
\end{aligned}
\end{equation}
Thus, we must solve the equations
\begin{equation}\label{eq:inducedmetric}
\begin{aligned}
    \gamma_{\tau\tau} &= g_{rr}\dot{r}^{2} - r^2 \dot{t}^2 = - r_c^2/L^2, \\
    \gamma_{\theta\theta} &= g_{\theta\theta} = r_c^2,
\end{aligned}
\end{equation}
where dots are derivatives with respect to $\tau$.
In fact we will only need $\dot{r}(0)$ and $\dot{\tau}(0)$.
Taking a derivative of the second equation and then setting $\tau=0$ gives the pair of equations
\begin{equation}
\begin{aligned}
    \dot{r}(0)^2 g_{rr}(0, r_c) -r_c^2 \dot{t}(0)^2&= - r_c^2/L^2,  \\
     \partial_{t} g_{\theta\theta}(0, r_c) \dot{t}(0) + \partial_r g_{\theta\theta}(0, r_c) \dot{r}(0) &= 0,
\end{aligned}
\end{equation}
or
\begin{equation}
\begin{aligned}
    \dot{r}(0)^2 B(r_c) -r_c^2 \dot{t}(0)^2&= - r_c^2/L^2, \\
     2 r_c K_{\theta\theta}(r_c) \dot{t}(0) + 2 r_c \dot{r}(0) &= 0,
\end{aligned}
\end{equation}
which are easily solved to give 
\begin{equation}
\begin{aligned}
    \dot{t}(0) &= \frac{ 1 }{ \sqrt{1 - B K_{\theta\theta}^2 r^{-2}L^2}  }\Big|_{r=r_c}, \\
    \dot{r}(0) &= -\frac{ K_{\theta\theta} }{ \sqrt{1 - B K_{\theta\theta}^2 r^{-2}L^2}  }\Big|_{r=r_c}, \\
\end{aligned}
\end{equation}
where we choose the branch where $\dot{t}>0$. Remembering from Eq.~\eqref{eq:Ksym} that
\begin{equation}
    \begin{aligned}\label{eq:Ktheta}
    K_{\theta\theta} = - \frac{ K_{rr}r^2 }{ (d-1) B(r) },
\end{aligned}
\end{equation}
we find that the term $B K_{\theta\theta}^2 r^{-2}$ at large $r$ behaves like
\begin{equation}
\begin{aligned}
    BK_{\theta\theta}^2r^{-2} \sim \frac{K_{rr}^2 r^2}{B} \sim K_{rr}^2 r^4 L^{-2}
\end{aligned}
\end{equation}
The solution of \eqref{eq:constraintEqsJ} gives that $\mathcal{K}(r) \sim \mathcal{O}(r^{-d})$, so we have
\begin{equation}
\begin{aligned}
    K_{rr} = B(r) \mathcal{K}(r) \sim \frac{ 1 }{ r^{d+2} },
\end{aligned}
\end{equation}
implying that $BK_{\theta\theta}^2r^{-2}\sim \mathcal{O}\left(r^{-2d}\right)$.
We thus find
\begin{equation*}
\begin{aligned}
    \dot{t}(0) &= 1 + \mathcal{O}\left(r_c^{-2d} \right), \\
    \dot{r}(0) &= \frac{ K_{rr} r^2 }{(d-1)B(r)  }\left(1 + \mathcal{O}\left(r^{-2d} \right) \right)\Big|_{r=r_c}.
\end{aligned}
\end{equation*}
In our spacetime coordinates, $\eta^a$ and $N^a$ reads
\begin{equation*}
\begin{aligned}
    \eta^a &= (\partial_{\tau})^a = \dot{t}(0)(\partial_{t})^a + \dot{r}(0)(\partial_{r})^a, \\
    N^{a} &= \frac{ 1 }{ \sqrt{B(r_c)} } (\partial_{r})^a,
\end{aligned}
\end{equation*}
giving finally that
\begin{equation}
\begin{aligned}
    \eta^a N_a &= \sqrt{B(r_c)}\dot{r}(0) = \frac{ K_{rr}r^2 }{ (d-1)\sqrt{B(r)} }\left[1 + \mathcal{O}\left(r^{-2d}
    \right) \right]\Big|_{r=r_c}.
\end{aligned}
\end{equation}
This implies
\begin{equation}
\begin{aligned}
    \frac{ \dd V[\Sigma] }{ \dd \tau } &=  \frac{\Omega_{k }
    r^{d+1} K_{rr} }{ (d-1) \sqrt{B(r)} } \left[ 1 + \mathcal{O}(r^{-2d}) \right]|_{r=r_c} \\
    &= \frac{\Omega_{k}r^{d+2}K_{rr}}{d-1}\left[1 + \mathcal{O}(r^{-2}, \omega(r) r^{-d}) \right]|_{r=r_c}
\end{aligned}
\end{equation}
This is clearly finite when $\omega(\infty)$ is finite. Furthermore $r^{d+2} K_{rr}\sim
\mathcal{O}(1)$ and so we have
\begin{equation}
\begin{aligned}
    \frac{ \dd \CV }{ \dd \tau } = \frac{ \Omega_{k} }{ G_N L (d-1) } \lim_{r \rightarrow \infty} r^{d+2}K_{rr}
    = \frac{ \Omega_{k} }{ G_N L (d-1) }\lim_{r \rightarrow \infty}r^{d}\mathcal{K}(r).
\end{aligned}
\end{equation}
Note that the above derivations holds upon the replacement of $\dd\Omega^2$ with $\dd \bm{x}^2$, 
and $\theta$ with one of the Cartesian boundary directions. 
The hyperbolic case requires minor modifications, but is not needed for us.

If we plug in the solution of Eq.~\eqref{eq:constraintEqsJ} into the expression for $\dot{\mathcal{C}}$, we find
\begin{equation}
\begin{aligned}
    \frac{ \dd \CV}{ \dd \tau }&= \frac{ \Omega_{k} }{ G_N L (d-1)  }\left[\mathcal{K}(r_0)r_0^{d} +
    \int_{r_0}^{\infty}\dd r r^{d}J_r(r) \right] \\
    &=  \frac{ 1 }{ G_N L (d-1)  }\left[r_0 \int_{\partial \Sigma_{\rm out}(r_0)} r^{\alpha}r^{\beta}K_{\alpha\beta} + 8\pi G_N
    \int_{\Sigma_{\rm out}(r_0)} r T_{\alpha\beta}n^{\alpha}r^{\beta} \right], \\
\end{aligned}
\end{equation}
If we take $r_0$ to be at an outermost stationary surface, $r_0=r_{\rm throat}$, where $\mathcal{K}(r_0) =
\sqrt{2\theta_{k}\theta_{\ell}}$, we find the expression
\begin{equation}
    \begin{aligned}
    \frac{ \dd \CV }{ \dd \tau }
    &=  \frac{ 1 }{ G_N L (d-1)  }\left[r_{\rm throat} \int_{\partial \Sigma_{\rm out}(r_{\rm throat})}\sqrt{2\theta_{k}\theta_{\ell}} + 8\pi G_N
    \int_{\Sigma_{\rm out}(r_{\rm throat})} r T_{\alpha\beta}n^{\alpha}r^{\beta} \right]. \\
\end{aligned}
\end{equation}

\toclesslab\subsection{A simple purification}{sec:diamondappendix}
\subsubsection*{A spherically symmetric simplification}
Consider a spherically symmetric bulk causal diamond $D$ anchored at a bulk subregion $R$ that is just a
boundary sphere. $D$ could be an entanglement wedge (intersected with the Wheeler-de-Witt patch of
$R$ -- i.e. the domain of dependence is taken in the strict bulk sense), but need not be. 
Let $\Sigma$ be the maximal volume slice of $D$. Then the edge of $D$, which equals $\partial \Sigma$, has the intrinsic metric of a
sphere. We now want to find a surface $\Sigma_0$ and initial data $(\Sigma_0, h, K_{ab})$ such that $\Sigma\cup
\Sigma_0$ (1) is complete, (2) is an extremal hypersurface, (3) satisfies our energy conditions and (4) has minimal
volume consistent with (1)--(3). Let us only consider
$\Sigma_0$ to be spherically symmetric 
-- in principle there might be a $\Sigma_0$ with even less volume satisfying our above conditions which breaks spherical
symmetry. 

If it is compatible with energy conditions to have $\Sigma_0$ be compact, i.e.\ without a conformal
boundary, then this is always the correct choice, since adding a second conformal boundary
introduces new UV divergences. Let us thus assume $\Sigma_0$ is one-sided and later return to when
this is consistent. Our computations in Sec.~\ref{sec:withsym} showing that pure AdS was the least complex 
space did not rely on $r=\infty$ being the upper
bound of the volume integral. Let $\Gamma$ be some compact spherically symmetric extremal spacelike
manifold with a spherical boundary of radius $r$. Then
        \begin{equation}
        \begin{aligned}
            V[\Gamma] &= \vol[S^{d-1}]\int_{0}^{r} \frac{ \dd \rho \rho^{d-1} }{ \sqrt{1 + \rho^2/L^2 - \frac{ \omega(\rho)
            }{ \rho^{d-2} }} }\\
            &\geq \vol[S^{d-1}] \int_{0}^{r} \frac{ \dd \rho \rho^{d-1} }{ \sqrt{1 + \rho^2/L^2} } \\
            &= V[\Sigma_0] = \frac{
            \vol[S^{d-1}]r^{d} }{ d }{}_2 F_{1}\left(\frac{ 1 }{ 2 }, \frac{ d }{ 2 }, \frac{ 2+d }{ 2 }, - \frac{ r^2 }{
            L^2 } \right),
        \end{aligned}
        \end{equation}
        where $\Sigma_0$ is the ball in the hyperbolic plane with boundary of area radius $r$. 
        Here we used that $\omega(r)\geq 0$,
        which is shown in Lemma~\ref{lem:minorlem}.
        Thus, of all spherically symmetric manifolds with boundary of area radius $r$,
        a ball in hyperbolic space has the least volume -- assuming the WCC. Thus, we should just choose the completion $\Sigma_0$ of
        $\Sigma$ to be a ball in hyperbolic space, provided that the energy shell induced at the
        junction between $\Sigma$ and $\Sigma_0$ is compatible with the WCC
        and the NEC, which we return to in a moment. See Fig.~\ref{fig:simplified} for an
        illustration of the new spacetime.
\begin{figure}
\centering
\includegraphics[width=0.3\textwidth]{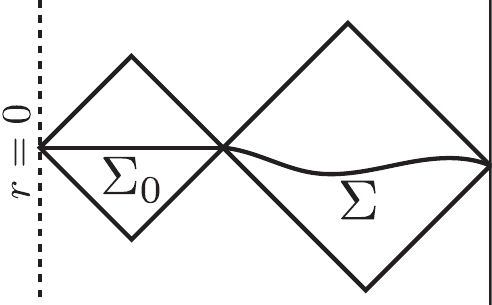}
    \caption{Construction of the simplified spacetime. The slice $\Sigma_0$ has intrinsic and extrinsic
    geometry of a static slice of AdS.}
\label{fig:simplified}
\end{figure}

Now let us prove that any other spacetime with spherical symmetry has higher complexity: 
Let $(\tilde{M}, \tilde{g})$ be any spherically symmetric spacetime containing $D$, and let
$\tilde{R}$ be a boundary Cauchy slice containing $R$.
If $R$ is a strict subset of $\tilde{R}$, then the maximal volume slice anchored at $\tilde{R}$ trivially has larger volume due to
additional UV-divergences. Thus, assume $R=\tilde{R}$, meaning that $(\tilde{M},
\tilde{g})$ is one-sided. Let $\Gamma$ be an extremal Cauchy slice of the inner wedge $I_W[\partial \Sigma]$  of $\partial
        \Sigma$ in $(\tilde{M}, \tilde{g})$ -- see Fig.~\ref{fig:proof_fig}. 
        \begin{figure}
        \centering
        \includegraphics[width=0.4\textwidth]{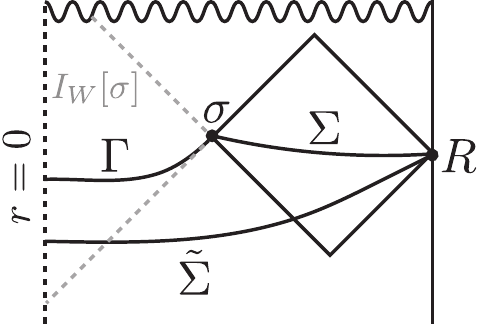}
            \caption{A spherically symmetric spacetime $(\tilde{M}, \tilde{g})$ containing $D[\Sigma]$. The maximal volume slice $\tilde{\Sigma}$
            must have volume greater or equal to \eqref{eq:CPcomp}.  }
        \label{fig:proof_fig}
        \end{figure}
        Let $\tilde{\Sigma}$ be the maximal volume slice of $(\tilde{M}, \tilde{g})$ anchored at
        $R$. Since $\Sigma \cup \Gamma$ is a complete but non-extremal slice due to the kink at the joining, we
        have that its volume is less than the volume of $\tilde{\Sigma}$:
        \begin{equation}
        \begin{aligned}
             V[\tilde{\Sigma}]  \geq V[\Sigma] + V[\Gamma].
        \end{aligned}
        \end{equation}
        But $\Gamma$ is a spherically symmetric compact manifold with boundary being a sphere, so it
        has greater or equal volume than the ball $\Sigma_0$ of the hyperbolic plane with the same
        area of its boundary. Thus
        \begin{equation}
        \begin{aligned}
            V[\tilde{\Sigma}]  \geq V[\Sigma] + V[\Gamma] \geq V[\Sigma] + V[\Sigma_0],
        \end{aligned}
        \end{equation}
        completing our proof.

\subsubsection*{Energy conditions at the junction}
By trying to glue $\Sigma$ to a subset of the hyperbolic plane $\Sigma_{0}$, the thing that can go
wrong is that we might be forced to violate the WCC or the NEC, since a
junction supports a distributional Ricci tensor that might or might not respect our energy
conditions. Below we show that the WCC (and thus the NEC) is preserved by the energy shell at the junction if $K_{\alpha\beta}=0$ on $\Sigma_0$ and if
\begin{equation}\label{eq:thetacondition}
    \sqrt{2} \theta_{k}[\partial \Sigma] \leq H_{\Sigma_0}[\partial \Sigma_0],
    \qquad \sqrt{2} \theta_{\ell}[\partial \Sigma] \geq - H_{\Sigma_0}[\partial \Sigma_0],
\end{equation}
where $k, \ell$ are the future directed normals to $\partial \Sigma$ with $k$ pointing towards $R$, and with
\begin{equation}\label{eq:canNorm}
\begin{aligned}
    k \cdot n = \ell \cdot n = - \frac{ 1 }{ \sqrt{2} },
\end{aligned}
\end{equation}
where $n^a$ is the future unit normal to $\Sigma$.
Here $H_{\Sigma_0}[\partial \Sigma_0]$ is the mean curvature in $\Sigma_0$ with respect to the normal of $\partial \Sigma_0$ 
pointing toward $\Sigma$.
For the case of spherical symmetry the conditions read
\begin{equation}
        \begin{aligned}\label{eq:simpsphere}
        \theta_{k}[\partial \Sigma] \leq \frac{ d-1 }{\sqrt{2} }\sqrt{\frac{ 1 }{ r^2 } + \frac{ 1 }{ L^2 }}, \qquad
        \theta_{\ell}[\partial \Sigma] \geq
        - \frac{ d-1 }{ \sqrt{2} }\sqrt{\frac{ 1 }{ r^2 } + \frac{ 1 }{ L^2 }},
\end{aligned}
\end{equation}
These are clearly satisfied for an HRT surface, since $\theta_k[\partial \Sigma] = \theta_{\ell}[\partial \Sigma]=0$.

Let us now derive these conditions. Consider a surface $\sigma$ contained in an extremal hypersurface $\Sigma$. Consider the null vectors
\begin{equation}
    k^a = \frac{ 1 }{ \sqrt{2} }(n^a + r^a), \quad \ell^a = \frac{ 1 }{ \sqrt{2} }(n^a - r^a),
\end{equation}
where we $n^a$ is the future timelike unit normal to $\Sigma$ and $r^a$ the spacelike outwards normal to
$\sigma$ contained in $\Sigma$. Letting $P^{ab}$ be the projector $\Sigma$ and $h^{ab}$ the projector on $\sigma$, we find
\begin{equation}\label{eq:thetaslice}
\begin{aligned}
    \sqrt{2} \theta_{k}
                &= h^{ab}\nabla_{a}(n_b + r_{b}) 
                = (P^{ab}-r^a r^b) (\nabla_{a}n_b  + \nabla_a r_{b}) \\
                &= K + P^{ab}\nabla_a r_b - r^a r^b \nabla_a n_b 
                & =K  + H_{\Sigma}[\sigma] - r^{\alpha}r^{\beta} K_{\alpha \beta}.
\end{aligned}
\end{equation}
A similar calculation shows that 
\begin{equation}\label{eq:thetaslice2}
\sqrt{2} \theta_{\ell} = K - H_{\Sigma}[\sigma]- r^{\alpha}r^{\beta} K_{\alpha \beta}.
\end{equation}
Thus,
\begin{equation}\label{eq:Htheta}
    H_{\Sigma}[\sigma] = \frac{ 1 }{ \sqrt{2} }\left[\theta_{k_{\Sigma}}- \theta_{\ell_{\Sigma}}\right].
\end{equation}
If $K_{\alpha\beta}=0$, which is the case for a constant-$t$ slice of AdS, then
\begin{equation}
\begin{aligned}
\theta_{k} = - \theta_{\ell},
\end{aligned}
\end{equation}
and 
\begin{equation}
    H_{\Sigma}[\sigma] = \sqrt{2}\theta_{k} = -\sqrt{2} \theta_{\ell}.
\end{equation}

Let us now consider the junction conditions. In \cite{Poisson}, it is shown that for a null junction between two spacetime regions $M^+$ and
$M^-$ generated by $k^a$, we have 
\begin{equation}
    R_{\ell \ell}|_{\rm singular} = \frac{ \delta(\tau) }{ -k\cdot n }\left(\theta_{\ell}^{-} - \theta_{\ell}^+ \right),
\end{equation}
where $n^a$ is tangent to a time-like congruence crossing the junction and $\tau$ the proper time from the
junction along the congruence generated by $n^a$. For this formula to be valid we must have
\begin{equation}
    [n\cdot k] \equiv n\cdot k|_{+} - n\cdot k|_{-} = 0.
\end{equation}
The region $M^+$ is the region to the future of the junction and $M^-$ to the past. 
Taking $n^a$ to be future directed, and taking $k^a$ and to be future-directed and to point towards
$R$, we have that the $M^-$-region contains $\Sigma$. Thus the NEC demands
\begin{equation}
    \theta_{\ell}[\partial \Sigma] - \theta_{\ell}[\partial \Sigma_0] \geq 0.
\end{equation}

Similarly, for a null-junction generated by $\ell$ we find
\begin{equation}
    R_{kk}|_{\rm sing} = \frac{ \delta(\tau) }{ -\ell \cdot n }\left(\theta_{k}^{-} - \theta_{k}^+ \right),
\end{equation}
where $[\ell \cdot n] =0$.
This time however, the $M^-$-region contains $\Sigma_{0}$, and so the NEC demands
\begin{equation}
    \theta_{k}[\partial \Sigma_0] - \theta_{k}[\partial \Sigma] \geq 0.
\end{equation}

Finally let us now consider what (non-null) energy shock is required. We do this by directly applying the spacelike junction conditions in the
Riemannian setting, forgetting about the ambient spacetime. Define 
\begin{equation}
    [A] = A|_{\rm out} - A|_{\rm in}.
\end{equation}
The Riemannian Einstein tensor is given by \cite{Poisson}
\begin{equation}
    G_{\alpha\beta}= \text{non-singular} + \delta(s) S_{\alpha \beta},
\end{equation}
where $s$ is the proper distance from $\sigma$ and 
\begin{equation}
    S_{\alpha\beta} = -[H_{\alpha\beta}] + [H] h_{\alpha\beta},
\end{equation}
where we remind that $h_{\alpha\beta}$ is the intrinsic metric on $\Sigma$ and $H_{\alpha\beta}$ the extrinsic curvature of $\sigma$
in $\Sigma$ with respect the outwards normal $r^\alpha$ (pointing towards $R$). This computation is purely geometric and does not assume Einstein's equations.
Taking the trace, we thus find
\begin{equation}
    R(h)\left(1 - \frac{ d }{ 2 } \right)  = \text{non-singular} + \delta(s) (d-2) [H].
\end{equation}
Intergrating $R$ on a small curve tangent $r^a$ accross $\sigma$, we find
\begin{equation}
    \int_{\delta}^{-\delta} \dd s R(h) = -2[H] + \mathcal{O}(\delta).
\end{equation}
Applying the Gauss-Codazzi equation, this gives
\begin{equation}
    \int_{-\delta}^{\delta} \dd \ell (2 \mathcal{E}+ K^{ab} K_{ab}) = -2[H] + \mathcal{O}(\delta).
\end{equation}
Thus, for an arbitrarily small $\delta>0$ the WCC dictates that
\begin{equation}
    \int_{-\delta}^{\delta} \dd \ell \mathcal{E}= H|_{\rm in}- H|_{\rm out} - \int \frac{ 1 }{ 2
    }K^{ab} K_{ab} \geq 0.
\end{equation}
It appears likely that $K_{ab}K^{ab}$ should not make any contribution to this integra in general, since $D_{\alpha}K^{\alpha\beta}=J^{\beta}$ should be solvable with a Green's function 
that smears the potential delta function in $J$, causing $K_{\alpha\beta}$ to be merely discontinous rather than distributional, and so 
\begin{equation}
H|_{\rm in} \geq H_{\rm out}.
\end{equation}
should be a sufficient condition for a positive energy shell. For spherical, planar and hyperbolic symmetry we can
explicitly check that this indeed is exactly what happens. By formula \eqref{eq:Htheta} we see that the constraint on $H$ reduces to
\begin{equation}
    \theta_{k}[\partial \Sigma_0] - \theta_{\ell}[\partial \Sigma_0] \geq \theta_{k}[\partial \Sigma] - \theta_{\ell}[\partial \Sigma],
\end{equation}
which is true when conditions on the NEC holds.

\bibliographystyle{jhep}
\bibliography{all}

\end{document}